\documentclass[11pt]{article}
\usepackage[utf8]{inputenc} 
\usepackage[T1]{fontenc}
\usepackage{lmodern}

\usepackage{geometry}
\usepackage{microtype}  
\usepackage{array}
\usepackage{multirow}
\usepackage{amsmath} 
\usepackage{amssymb}
\usepackage{amsthm} 
\usepackage{thmtools}

\usepackage[procnumbered,ruled,vlined,linesnumbered]{algorithm2e}
 
\SetCommentSty{mycommfont} 

\usepackage{xcolor}  
\usepackage{xspace}
\usepackage{enumitem}    
\usepackage{tikz}   

\usepackage{comment} 

\usepackage{graphicx} 
\usepackage{caption} 
\usepackage{subcaption}  

\usepackage{authblk} 

\def\AdvCite{True} 

\ifdefined\AdvCite

\usepackage[
style=alphabetic, 
backref=true,
doi=false,
url=false,
maxcitenames=3,
mincitenames=3,
maxbibnames=10,
minbibnames=10,
backend=bibtex8,
sortlocale=en_US
]{biblatex}

\usepackage[colorlinks]{hyperref} 
\usepackage[capitalise]{cleveref}

\DeclareCiteCommand{\citem}
{}
{\mkbibbrackets{\bibhyperref{\usebibmacro{postnote}}}}
{\multicitedelim}
{}

\renewcommand*{\multicitedelim}{\addcomma\space}

\AtEveryCitekey{\clearfield{note}}

\newcommand{\myhref}[1]{%
	\iffieldundef{doi}
	{\iffieldundef{url}
		{#1}
		{\href{\strfield{url}}{#1}}}
	{\href{http://dx.doi.org/\strfield{doi}}{#1}}%
}

\DeclareFieldFormat{title}{\myhref{\mkbibemph{#1}}}
\DeclareFieldFormat
[article,inbook,incollection,inproceedings,patent,thesis,unpublished]
{title}{\myhref{\mkbibquote{#1\isdot}}}

\makeatletter
\AtBeginDocument{%
	\newlength{\temp@x}%
	\newlength{\temp@y}%
	\newlength{\temp@w}%
	\newlength{\temp@h}%
	\def\my@coords#1#2#3#4{%
		\setlength{\temp@x}{#1}%
		\setlength{\temp@y}{#2}%
		\setlength{\temp@w}{#3}%
		\setlength{\temp@h}{#4}%
		\adjustlengths{}%
		\my@pdfliteral{\strip@pt\temp@x\space\strip@pt\temp@y\space\strip@pt\temp@w\space\strip@pt\temp@h\space re}}%
	\ifpdf
	\typeout{In PDF mode}%
	\def\my@pdfliteral#1{\pdfliteral page{#1}}
	\def\adjustlengths{}%
	\fi
	\ifxetex
	\def\my@pdfliteral #1{}
	\def\adjustlengths{\setlength{\temp@h}{-\temp@h}\addtolength{\temp@y}{1in}\addtolength{\temp@x}{-1in}}%
	\fi%
	\def\Hy@colorlink#1{%
		\begingroup
		\ifHy@ocgcolorlinks
		\def\Hy@ocgcolor{#1}%
		\my@pdfliteral{q}%
		\my@pdfliteral{7 Tr}
		\else
		\HyColor@UseColor#1%
		\fi
	}%
	\def\Hy@endcolorlink{%
		\ifHy@ocgcolorlinks%
		\my@pdfliteral{/OC/OCPrint BDC}%
		\my@coords{0pt}{0pt}{\pdfpagewidth}{\pdfpageheight}%
		\my@pdfliteral{F}
		\my@pdfliteral{EMC/OC/OCView BDC}%
		\begingroup%
		\expandafter\HyColor@UseColor\Hy@ocgcolor%
		\my@coords{0pt}{0pt}{\pdfpagewidth}{\pdfpageheight}%
		\my@pdfliteral{F}
		\endgroup%
		\my@pdfliteral{EMC}%
		\my@pdfliteral{0 Tr}
		\my@pdfliteral{Q}%
		\fi
		\endgroup
	}%
} 
\makeatother

\else

\usepackage[ocgcolorlinks]{hyperref} 
\usepackage{cleveref}

\fi 

\allowdisplaybreaks[1]

\makeatletter
\g@addto@macro\bfseries{\boldmath}
\makeatother

\makeatletter
\g@addto@macro\mdseries{\unboldmath}
\g@addto@macro\normalfont{\unboldmath}
\g@addto@macro\rmfamily{\unboldmath}
\g@addto@macro\upshape{\unboldmath}
\makeatother

\renewcommand{\paragraph}[1]{\medskip\noindent{\bf #1}\xspace}

\colorlet{DarkRed}{red!50!black} 
\colorlet{DarkGreen}{green!50!black}
\colorlet{DarkBlue}{blue!50!black}

\hypersetup{
	linkcolor = DarkRed,
	citecolor = DarkGreen,
	urlcolor = DarkBlue,
	bookmarks = true,
	bookmarksnumbered = true,
	linktocpage = true 
}

\declaretheorem[numberwithin=section]{theorem}
\declaretheorem[numberlike=theorem]{lemma}
\declaretheorem[numberlike=theorem]{proposition}
\declaretheorem[numberlike=theorem]{corollary}
\declaretheorem[numberlike=theorem]{claim}
\declaretheorem[numberlike=theorem]{observation}

\declaretheorem[numberlike=theorem]{assumption}

\crefname{algorithm}{Algorithm}{Algorithms}
\Crefname{algorithm}{Algorithm}{Algorithms}

\theoremstyle{definition}
\declaretheorem[numberlike=theorem]{definition}

\DontPrintSemicolon
\SetKw{KwAnd}{and}
\SetProcNameSty{textsc}
\SetFuncSty{textsc}

\newcommand{\ot}{\tilde{O}}

\newcommand{\ignore}[1]{}

\newcommand{\cC}{\mathcal{C}\xspace}

\newcommand{\cL}{\mathcal{L}\xspace}

\newcommand{\poly}{\operatorname{poly}} 
\newcommand{\polylog}{\operatorname{polylog}}

\newcommand{\vol}{\operatorname{vol}} 
\newcommand{\blockingflow}{\operatorname{BinaryBlockingFlow}} 
\newcommand{\VC}{\operatorname{VC}} 
\newcommand{\dmin}{\deg_{\min}^{\out}}

\newcommand{\rank}{\operatorname{rank}} 

\newcommand{\out}{\operatorname{out}}  
\newcommand{\textin}{\operatorname{in}}   
\newcommand{\textflow}{\operatorname{flow}}   

\ifdefined\ShowComment
\def\thatchaphol#1{\marginpar{$\leftarrow$\fbox{D}}\footnote{$\Rightarrow$~{\sf\textcolor{purple}{#1 --Thatchaphol}}}}
\def\danupon#1{\marginpar{$\leftarrow$\fbox{D}}\footnote{$\Rightarrow$~{\sf\textcolor{orange}{#1 --Danupon}}}}
\def\sorrachai#1{\marginpar{$\leftarrow$\fbox{D}}\footnote{$\Rightarrow$~{\sf\textcolor{green}{#1 --Shen}}}}
\def\note#1{#1}

\else
\def\thatchaphol#1{}
\def\danupon#1{}
\def\shen#1{}
\def\sorrachai#1{}
\def\note#1{} 

\fi

\title{Breaking Quadratic Time for Small Vertex Connectivity and an Approximation Scheme}

\author[1]{Danupon Nanongkai}
\author[2]{Thatchaphol Saranurak\thanks{Works partially done while at KTH Royal Institute of Technology, Sweden.}}
\author[3]{Sorrachai Yingchareonthawornchai\thanks{Works partially done while at KTH Royal Institute of Technology, Sweden.}}
\affil[1]{KTH Royal Institute of Technology, Sweden}
\affil[2]{Toyota Technological Institute at Chicago, USA}
\affil[3]{Michigan State University, USA, and Aalto University, Finland}

\date{} 
\begin{document}

	\begin{titlepage}   
		\maketitle
		\pagenumbering{roman}   
		
\begin{abstract}
Vertex connectivity a classic extensively-studied problem. Given an integer $k$, its goal is to decide if an $n$-node $m$-edge graph can be disconnected by removing $k$ vertices. 
Although a linear-time algorithm was postulated since 1974 [Aho, Hopcroft, and Ullman], and despite its sibling problem of edge connectivity being resolved over two decades ago [Karger STOC'96], so far no vertex connectivity algorithms are faster than $O(n^2)$ time even for $k=4$ and $m=O(n)$. 
In the simplest case where $m=O(n)$ and $k=O(1)$, the $O(n^2)$ bound dates five decades back to [Kleitman IEEE Trans. Circuit Theory'69].
For higher $m$, $O(m)$ time is known for $k\leq 3$ [Tarjan FOCS'71; Hopcroft, Tarjan SICOMP'73], the first $O(n^2)$ time is from [Kanevsky, Ramachandran, FOCS'87]  for $k=4$ and from [Nagamochi, Ibaraki, Algorithmica'92] for $k=O(1)$. 
For general $k$ and $m$, the best bound is $\tilde{O}(\min(kn^2, n^\omega+nk^\omega))$ [Henzinger, Rao, Gabow FOCS'96; Linial, Lov{\'{a}}sz, Wigderson FOCS'86] where $\tilde O$ hides polylogarithmic terms and $\omega<2.38$ is the matrix multiplication exponent.

In this paper, we present a randomized Monte Carlo algorithm with $\tilde{O}(m+k^{7/3}n^{4/3})$ time for any $k=O(\sqrt{n})$. This 
gives the {\em first subquadratic time} bound for any $4\leq k \leq o(n^{2/7})$ (subquadratic time refers to $O(m)+o(n^2)$ time) and improves all above classic bounds for all $k\le n^{0.44}$.  
We also present a new randomized Monte Carlo $(1+\epsilon)$-approximation algorithm that is strictly faster than the previous Henzinger's 2-approximation algorithm [J. Algorithms'97] and all previous exact algorithms.
The story is the same for the {\em directed} case, where our exact  $\tilde{O}( \min\{km^{2/3}n, km^{4/3}\}  )$-time for any $k = O(\sqrt{n})$ and  $(1+\epsilon)$-approximation algorithms improve all previous classic bounds. Additionally, our algorithm is the first  approximation algorithm on directed graphs. 

The key to our results is to avoid computing single-source connectivity, which was needed by all previous exact algorithms and is not known to admit $o(n^2)$ time. Instead, we design the first {\em local algorithm} for computing vertex connectivity; without reading the whole graph, our algorithm can find a separator of size at most $k$ or certify that there is no separator of size at most $k$ ``near'' a given seed node.
\end{abstract}
		\newpage   
		\setcounter{tocdepth}{2}  
		\tableofcontents    
	\end{titlepage}

	\newpage         
	\pagenumbering{arabic}           
	
	\section{Introduction}

Vertex connectivity is a central concept in graph theory. The vertex
connectivity $\kappa_{G}$ of a graph $G$ is the minimum number
of the nodes needed to be removed to disconnect some remaining node
from another remaining node.  (When $G$ is directed, this means that there is no directed path from some node $u$ to some node $v$ in the remaining graph.)

Since 1969, there has been a long line
of research on efficient algorithms \cite{Kleitman1969methods,Podderyugin1973algorithm,EvenT75,Even75,Galil80,EsfahanianH84,Matula87,BeckerDDHKKMNRW82,LinialLW88,CheriyanT91,NagamochiI92,CheriyanR94,Henzinger97,HenzingerRG00,Gabow06,Censor-HillelGK14}
for \emph{deciding $k$-connectivity }(i.e. deciding if $\kappa_{G}\ge k$)
or \emph{computing the connectivity} $\kappa_{G}$ (see
\Cref{tab:compare} for details).
For the undirected case,  Aho, Hopcroft, and Ullman \cite[Problem 5.30]{AhoHU74} conjecture in 1974 that there exists an
$O(m)$-time algorithm for computing $\kappa_{G}$ on a graph with
$n$ nodes and $m$ edges. However, no algorithms to date are faster than $O(n^2)$ time  even for $k=4$.

On undirected graphs, the first $O(n^2)$ bound for the simplest case, where $m=O(n)$ and $k=O(1)$, dates back to five decades ago:
Kleitman \cite{Kleitman1969methods} in 1969 presented an algorithm for deciding $k$-connectivity with running time
$O(kn\cdot\VC_{k}(n,m))$ where $\VC_{k}(n,m)$ is the time needed
for deciding if the minimum size $s$-$t$ vertex-cut is of size at
least $\kappa$, for fixed $s,t$. Although the running time bound
was not explicitly stated, it was known that $\VC_{k}(n,m)=O(mk)$
by Ford-Fulkerson algorithm \cite{ford1956maximal}. This gives $O(k^{2}nm)$ which is $O(n^{2})$
when $m=O(n)$ and $k=O(1)$, when we plug in the 1992 result of Nagamochi and Ibaraki \cite{NagamochiI92}.  
Subsequently, Tarjan~\cite{Tarjan72} and Hopcroft and Tarjan~\cite{HopcroftT73} presented $O(m)$-time algorithms when $k$ is $2$ and $3$ respectively.

All subsequent works 
improved Kleitman's bound for larger $k$ and $m$, but none could break beyond $O(n^{2})$ time. For $k=4$ and any $m$, the first $O(n^2)$ bound was by Kanevsky and Ramachandran~\cite{KanevskyR91}. The first $O(n^2)$ for any $k=O(1)$ (and any $m$) was by Nagamochi and Ibaraki \cite{NagamochiI92}.
For general $k$ and $m$, the fastest running
times are $\tilde{O}(n^{\omega}+nk^{\omega})$ by Linial, Lov{\'a}sz and Wigderson \cite{LinialLW88} and $\tilde{O}(kn^{2})$ by Henzinger, Rao and Gabow \cite{HenzingerRG00}. Here, $\tilde O$ hides $\polylog(n)$ terms, and $\omega$ is the matrix multiplication exponent. Currently, $\omega<2.37287$	\cite{Gall14a}.

For directed graphs, an $O(m)$-time algorithm is known only for $k\le2$ by Georgiadis \cite{Georgiadis10}. For  general $k$ and $m$, the fastest running times are $\tilde{O}(n^{\omega}+nk^{\omega})$
by Cheriyan and Reif \cite{CheriyanR94} and $\tilde{O}(mn)$ by Henzinger~et~al.~\cite{HenzingerRG00}.
All mentioned state-of-the-art algorithms for general $k$ and $m$, for both directed and undirected cases  \cite{LinialLW88,CheriyanR94,HenzingerRG00}, are randomized and correct with high probability. 
The fastest deterministic algorithm is by Gabow \cite{Gabow06} and has a slower running time. 
Some {\em approximation algorithms} have also been developed. The first is the deterministic $2$-approximation $O(\min\{\sqrt{n},k\}n^{2})$-time algorithm by Henzinger \cite{Henzinger97}. The second is the recent randomized $O(\log n)$-approximation $\tilde O(m)$-time algorithm by Censor-Hillel, Ghaffari, and Kuhn \cite{Censor-HillelGK14}. Both algorithms work only on undirected graphs.

Besides a few $O(m)$-time algorithms for $k\leq 3$, all previous exact algorithms could not go beyond $O(n^2)$ for a common reason: As a subroutine, they have to solve the following problem. For a pair of nodes $s$ and $t$, let $\kappa(s,t)$
denote the minimum number of nodes (excluding $s$ and $t$) required
to be removed so that there is no path from $s$ to $t$ in the remaining
graph. In all previous algorithms,
there is always some node $s$ such that these algorithms decide if
$\kappa(s,t)\ge k$ for all other nodes $t$ (and some algorithms
in fact computes $\kappa(s,t)$ for all $t$).  
We call this problem \emph{single-source	$k$-connectivity}. 
Until now, there is no $o(n^{2})$-time algorithm for this problem even when $k=O(1)$ and $m=O(n)$. 


\begin{table} 
	{\footnotesize
		
		\begin{tabular}{|>{\raggedright}p{0.17\textwidth}|>{\centering}p{0.25\textwidth}|>{\centering}p{0.25\textwidth}|>{\raggedright}p{0.27\textwidth}|}
			\hline 
			Reference & Directed & Undirected & Note\tabularnewline
			\hline 
			\hline 
			Trivial & \multicolumn{2}{c|}{$O(n^{2}\cdot\VC(n,m))$} & \tabularnewline
			\hline 
			\cite{Kleitman1969methods} & \multicolumn{2}{c|}{$O(kn\cdot\VC_{k}(n,m))$} & \tabularnewline
			\hline 
			\cite{Podderyugin1973algorithm,EvenT75} & \multicolumn{2}{c|}{$O(kn\cdot\VC(n,m))$} & \tabularnewline
			\hline 
			\cite{Even75} (cf. \cite{Galil80,EsfahanianH84,Matula87}) & \multicolumn{2}{c|}{$O((k^{2}+n)\cdot\VC_{k}(n,m))$} & \tabularnewline
			\hline 
			\cite{BeckerDDHKKMNRW82} & \multicolumn{2}{c|}{$\tilde{O}(n\cdot\VC(n,m))$} & Monte Carlo\tabularnewline
			\hline 
			\multirow{2}{0.2\textwidth}{\cite{LinialLW88} (\cite{CheriyanR94} for the directed case)} & \multicolumn{2}{c|}{$O((n^{\omega}+nk^{\omega})\log n)$} & Monte Carlo\tabularnewline
			\cline{2-4} 
			& \multicolumn{2}{c|}{$O((n^{\omega}+nk^{\omega})k)$} & Las Vegas\tabularnewline
			\hline 
			\cite{NagamochiI92,CheriyanT91} & - & $O(k^{3}n^{1.5}+k^{2}n^{2})$ & \tabularnewline
			\hline 
			\cite{Henzinger97} & - & $O(\min\{\sqrt{n},k\}n^{2})$ & $2$-approx.\tabularnewline
			\hline 
			\multirow{2}{0.2\textwidth}{\cite{HenzingerRG00}} & $O(mn\log n)$ & $O(kn^{2}\log n)$ & Monte Carlo\tabularnewline
			\cline{2-4} 
			& $O(\min\{n,k^{2}\}km+mn)$ & $O(\min\{n,k^{2}\}k^{2}n+\kappa n^{2})$ & \tabularnewline
			\hline 
			\cite{Gabow06} & $O(\min\{n^{3/4},k^{1.5}\}km+mn)$ & $O(\min\{n^{3/4},k^{1.5}\}k^{2}n+kn^{2})$ & \tabularnewline
			\hline 
			\cite{Censor-HillelGK14} & - & $\tilde{O}(m)$ & Monte Carlo, $O(\log n)$-approx.\tabularnewline
			\hline 
			\multirow{2}{0.2\textwidth}{\textbf{This
                                                                        paper}
                                                                        }
                  & $\tilde{O}(\min\{km^{2/3}n, km^{4/3}\}) $ & $\tilde{O}(m+k^{7/3}n^{4/3})$& Monte Carlo, for $k\le\sqrt{n}$\tabularnewline
			\cline{2-4} 
& $\tilde{O}(t_{\text{directed}}) $ & $\tilde{O}(t_{\text{undirected}})$ & Monte Carlo, $(1+\epsilon)$-approx. \tabularnewline
			\hline 
		\end{tabular}
		
	}
	
	\caption{\label{tab:compare}List of running time $T(k,n,m)$ of previous algorithms
		on a graph $G$ with $n$ nodes and $m$ edges for deciding if the
		vertex connectivity $\kappa_{G}\ge k$. $\protect\VC(n,m)$ is the
		time needed for finding the minimum size $s$-$t$ vertex cut for
		fixed $s,t$. $\protect\VC_{k}(n,m)$ is the time needed for either
		certifying that the minimum size $s$-$t$ vertex cut is of size at
		least $k$, or return such cut. Currently, $\protect\VC_{k}(n,m)=O(\min\{k,\sqrt{n}\}m)$.
		If $\kappa_{G}\le k$, all algorithms above can compute $\kappa_{G}$,
		and most algorithms (except those in \cite{LinialLW88,CheriyanR94})
		also return a corresponding
                separator. $t_{\text{directed}} = \poly(1/\epsilon) \min(T_{\textflow}(k,m,n), n^{\omega}) $ , and
                $t_{\text{undirected} } =  m+\poly(1/\epsilon) \min(
                k^{4/3} n^{4/3}, k^{2/3} n^{5/3+o(1)}, n^{3+o(1)}/k,
                n^{\omega} )$ where $T_{\textflow}$ is defined in \Cref{eq:exact-directed-time-intro}.}
\end{table}

\subsection{Our Results}

In this paper, we present the first algorithms that break the $O(n^2)$ bound on both undirected and undirected graphs, when $k$ is small. More precisely:

\begin{theorem}
	\label{thm:exact} There are randomized (Monte Carlo)
	algorithms that take as inputs an $n$-node $m$-edge graph $G=(V,E)$  and an integer $k=O(\sqrt{n})$, and can decide w.h.p.\footnote{We say that an event holds with high probability (w.h.p.) if it holds with probability at least $1-1/n^c$, where $c$ is an arbitrarily large constant.}
	if $\kappa_{G}\ge k$.  If $\kappa_{G}<k$,
	then the algorithms also return the corresponding separator $S\subset V$,
	i.e. a set $S$ where $|S|=\kappa_{G}$ and $G[V-S]$ is not
	connected if $G$ is undirected and not strongly connected if
        $G$ is directed. The algorithm takes
        $\tilde{O}(m+k^{7/3}n^{4/3})$  and $\tilde{O}( 
        \min(km^{2/3}n, km^{4/3} ))$ time on undirected and directed
	graphs, respectively.
\end{theorem}

Our bounds are the {\em first} $o(n^2)$ for the range $4\leq k \leq o(n^{2/7})$ on undirected graphs and range $3\leq k \leq o(n/m^{2/3})$ on directed graphs. 
Our algorithms are combinatorial, meaning that they do not rely on fast matrix multiplication. 
For all range of $k$ that our algorithms support,
i.e. $k=O(\sqrt{n})$, our algorithms 
improve upon the previous best combinatorial algorithms by Henzinger~et~al.~\cite{HenzingerRG00}, which take
time $\tilde{O}(kn^{2})$ on undirected graphs and $\tilde{O}(mn)$
on directed graphs\footnote{As $k\le\sqrt{n}$ and $m\ge nk$, we have
  $k\le m^{1/3}$. So $km^{2/3}n\le mn$.}.
Comparing with the $\tilde{O}(n^{\omega}+nk^{\omega})$ bound based on algebraic techniques by Linial~et~al.~\cite{LinialLW88}
and Cheriyan and Reif~\cite{CheriyanR94}, our algorithms are faster on undirected graphs
when $k\le n^{3\omega/7-4/7}\approx n^{0.44}$.
For directed graphs, our algorithm is faster where the range $k$
depends on graph density.  For example, consider the interesting case the graph is sparse but can still be $k$-connected which is when $m = O(nk)$. 
Then ours is faster than \cite{CheriyanR94} for any $k\le n^{0.44}$ like the undirected case.
However, in the dense case when $m = \Omega(n^2)$, ours is faster than
\cite{CheriyanR94} for any $k \leq n^{\omega-7/3}\approx n^{0.039}$.  
%

To conclude, our bounds are lower than all previous bounds when $4\leq
k \leq n^{0.44}$ for undirected graphs  and $3\leq k \leq n^{0.44}$ for directed sparse graphs (i.e. when $m = O(nk)$).
All these bounds \cite{HenzingerRG00,LinialLW88,CheriyanR94} have not been broken for over $20$ years.
In the simplest case where $m = O(n)$ and, hence $k= O(1)$, we break the 49-year-old $O(n^2)$ bound \cite{Kleitman1969methods}
down to $\tilde{O}(n^{4/3})$ for both undirected and directed graphs, respectively.

\danupon{SELF NOTE:  $nk^{\omega}$ kicks in when $nk^\omega\geq n^\omega$, which is when (approximately) $k\geq 0.578$}

\paragraph{Approximation algorithms.}
We can adjust the same techniques to get  $(1+\epsilon)$-approximate
$\kappa_{G}$ with faster
running time. In addition, we give another algorithm using a different technique that can $(1+\epsilon)$-approximate   $\kappa_{G}$ 
in $\tilde{O}(n^{\omega}/\epsilon^{2})$ time.

We define the function $T_{\text{flow}}(k,m,n)$ as 

\begin{equation} \label{eq:exact-directed-time-intro}
T_{\textflow}(k,m,n) = \left\{ \begin{array}{rl}
  \min(m^{4/3}, nm^{2/3}k^{1/2} , mn^{2/3+o(1)}/k^{1/3}, n^{7/3+o(1)}/k^{1/6} )     &\mbox{ if $ k \leq n^{4/5}$,} \\
   n^{3+o(1)}/k    &\mbox{ if $k > n^{4/5}$. }
  \end{array} \right. 
\end{equation}

\begin{theorem}[Approximation Algorithm]\label{thm:intro:approx}
There is a randomized (Monte Carlo)
algorithm that takes as input an $n$-node $m$-edge graph $G=(V,E)$ and w.h.p. outputs $\tilde{\kappa}$,
where $\kappa_{G}\le\tilde{\kappa}\le(1+\epsilon)\kappa_{G}$, in 
$\ot ( m+\poly(1/\epsilon) \min( k^{4/3} n^{4/3}, k^{2/3}
n^{5/3+o(1)}, n^{3+o(1)}/k, n^{\omega} )) = \ot(\min\{n^{2.2},n^{\omega}\})$ time
for undirected graph, and in $\ot (\poly(1/\epsilon)
\min(T_{\textflow}(k,m,n), n^{\omega}) ) =  \ot(\min\{n^{7/3},n^{\omega}\})$
time for directed graph where $T_{\textflow}(k,m,n)$ is defined in
\Cref{eq:exact-directed-time-intro}.  
The algorithm also returns a pair of nodes $x$ and $y$
	where $\kappa(x,y)=\tilde{\kappa}$. Hence, with additional $O(m\min\{\sqrt{n},\tilde{\kappa}\})$	time, the algorithm can compute the corresponding separator.
\end{theorem}

As noted earlier, previous algorithms achieve $2$-approximation in $O(\min\{\sqrt{n},k\}n^{2})$-time \cite{Henzinger97} and $O(\log n)$-approximation in $\tilde O(m)$ time \cite{Censor-HillelGK14}. For all possible values of $k$, our algorithms are strictly faster than the $2$-approximation algorithm of \cite{Henzinger97}. 

Our approximation algorithms are also strictly faster than all
previous exact algorithms with current matrix multiplication time (and
are never slower even if $\omega < 2.2$). 
 In particular,  even when 
 $\epsilon = 1/n^{\gamma}$ for small constant $\gamma>0$,
our algorithms are always polynomially faster than the exact algorithms
by \cite{HenzingerRG00} with running time $\tilde{O}(mn)$ and $\tilde{O}(kn^{2})$
on directed and undirected graphs, respectively. Compared with the bound $\tilde{O}(n^{\omega}+nk^{\omega})$
by \cite{LinialLW88} and \cite{CheriyanR94}, our bound for undirected
and directed graphs are $\ot(\min\{n^{2.2},n^{\omega}\})$ and
$\ot(\min\{n^{7/3}, n^\omega \})$, respectively for any density, which are less than current matrix multiplication time.   
%
Finally, note that the previous approximation algorithms \cite{Henzinger97,Censor-HillelGK14} only work on undirected graphs, while we also show algorithms on directed graphs.

\subsection{The Key Technique}

At the heart of our main result in \Cref{thm:exact} is a new \emph{local
	algorithm} for finding minimum vertex cuts. In general, we say that
an algorithm is \emph{local} if its running time does not depend
on the size of the whole input.
 
More concretely, let $G=(V,E)$ be a directed graph
where each node $u$ has out-degree $\deg^{\out}(u)$. Let $\dmin=\min_{u}\deg^{\out}(u)$
be the minimum out-degree. For any set $S\subset V$, the \emph{out-volume}
of $S$ is $\vol^{\out}(S)=\sum_{u\in S}\deg^{\out}(u)$ and the set
of \emph{out-neighbors} of $S$ is $N^{\out}(S)=\{v\notin S\mid(u,v)\in E\}$.
We show the following algorithm (see \Cref{thm:local-vertex-connectivity} for more
detail):

\begin{theorem}
	[Local vertex connectivity (informal)]\label{thm:local informal}
	There is a deterministic algorithm that takes as inputs
	a node $x$ in a graph $G$ and parameters $\nu$ and $k$ where $\nu,k$ are not too large, and in $\tilde{O}(\nu^{1.5}k)$ time 
	either 
	\begin{enumerate}
		\item returns a set $S\ni x$ where
		 $|N^{\out}(S)|\le k$,
		or 
		\item certifies that there is no set $S\ni x$ such that $\vol^{\out}(S)\le\nu$
		and $|N^{\out}(S)|\le k$. 
	\end{enumerate}
\end{theorem}

Our algorithm is the first local algorithm for finding small vertex
cuts (i.e. finding a small separator $N^{\out}(S)$). The algorithm
either finds a separator of size at most $k$, or certifies that no
separator of size at most $k$ exists ``near'' some node $x$. 
Our algorithm is exact in the sense that there is no gap on the cut size $k$ in the two cases.

Previously, there was a rich literature on local algorithms for finding
\emph{low conductance cuts}\footnote{The conductance of a cut $(S,V-S)$ is defined as $\Phi(S)=\frac{|E(S,V-S)|}{\min\{\vol(S),\vol(V-S)\}}$.},
which is a different problem from ours. The study was initiated by Spielman
and Teng \cite{SpielmanT04} in 2004. Since then, deep techniques
have been further developed, such as \emph{spectral-based} techniques\footnote{They are algorithms based on some random-walk or diffusion process.} 
(e.g.
\cite{SpielmanT13,AndersenCL06,AndersenL08,AndersenP09,GharanT12})
and newer \emph{flow-based} techniques \cite{OrecchiaZ14,HenzingerRW17,WangFHMR17,VeldtGM16}).
Applications of these techniques for finding low conductance cuts
are found in various contexts (e.g. balanced cuts \cite{SpielmanT13,SaranurakW19}),
edge connectivity \cite{KawarabayashiT15,HenzingerRW17}, and dynamically
maintaining expanders \cite{Wulff-Nilsen17,NanongkaiS17,NanongkaiSW17,SaranurakW19}).

It is not clear a priori that these previous techniques can be used
for proving \Cref{thm:local informal}. First of all, they were invented to
solve a different problem, and there are several small differences in
technical input-output constraints. More importantly, there is a conceptual difference, as follows. In most previous
algorithms, there is a ``gap'' between the two cases of the guarantees.
That is, if in one case the algorithm can return a cut $S\ni x$ whose conductance
is at most $\phi\in(0,1)$, then in the other case the algorithm can only guarantee that
there is no cut ``near'' $x$ with conductance $\alpha\phi$, for some $\alpha=o(1)$
(e.g. $\alpha =  O(\phi)$ or $O(1/\log n)$)\footnote{The algorithms from  \cite{KawarabayashiT15,HenzingerRW17} in fact do not 
	guarantee non-existence of some low conductance cuts in the second case, 
	but the guarantee is about min-cuts.}.

Because of these differences, not many existing techniques can be adapted to design a local algorithm for vertex connectivity. In fact, we are not aware of any spectral-based algorithms that can solve this problem, even when we can read the whole graph. 
Fortunately, it turns out that \Cref{thm:local informal} can be proved
by adapting some recent flow-based techniques. In general, a challenge in designing flow-based algorithms is to achieve the following goals simultaneously.
\begin{enumerate}
	\item Design some well-structured graph so that finding flows on this graph is useful for our application (proving \Cref{thm:local informal} in this case).
	We call such graph an \emph{augmented graph}.
	\item At the same time, design a local flow-based algorithm that is fast when running on the augmented graph.

\end{enumerate}

For the first task, the design of the augmented graph requires some careful choices (see \Cref{sub:overview flow} for the high-level ideas and \Cref{sec:augmented graph} for details).
For the second task, it turns out that previous
flow-based local algorithms \cite{OrecchiaZ14,HenzingerRW17,WangFHMR17,VeldtGM16}
can be adjusted to give useful answers for our applications when run on our augmented graph. 
However, these previous algorithms only give a slower running time of at least 
$\tilde{O}((\nu k)^{1.5})$.
To obtain the $\tilde{O}(\nu^{1.5}k)$ bound, we first speed up
Goldberg-Rao max flow algorithm \cite{GoldbergR98} from running time
$\tilde{O}(m\min\{\sqrt{m},n^{2/3}\})$ to $\tilde{O}(m\sqrt{n})$ when running on a graph with certain structure.
Then, we ``localize'' this algorithm in a similar manner
as in \cite{OrecchiaZ14}, which completes our second task (see \Cref{sub:overview flow} for more discussion).\footnote{Update: In the earlier version, we claimed that we could achieve a deterministic algorithm for the  $s$-$t$ vertex connectivity problem on weighted directed graphs as a by-product. This claim is not correct, but this does not affect our main result.}

Given the key local algorithm in \Cref{thm:local informal}, we obtain
\Cref{thm:exact,thm:intro:approx} by combining our local algorithms with other known techniques including
random sampling, Ford-Fulkerson algorithm, Nagamochi Irabaki's connectivity
certificate \cite{NagamochiI92} and convex embedding \cite{LinialLW88,CheriyanR94}.
We sketch how everything fits together in \Cref{sec:overview}. 

\subsection{Updates}
After this paper first appeared in STOC 2019 \cite{NanongkaiSY-stoc19},  new local vertex connectivity algorithms were independently developed in \cite{NanongkaiSY-preSODA20} and \cite{ForsterY-preSODA20} (see \cite{ForsterNSYY-soda20} for the merged version). These algorithms are simpler than our local algorithm. Combining them with our framework leads to a nearly-linear time algorithm for the vertex connectivity problem when the connectivity is $O(\poly\log n)$. The new algorithm is still slower than the algorithm in this paper when the input graph is directed and dense and $k$ is moderately large; for $(1+\epsilon)$-approxmiation, the new algorithm is still slower when $k$ is moderately large for both undirected and directed graphs. 

More recently, \cite{LiNPSY-stoc21} shows that vertex connectivity can be solved in roughly the time to compute a max-flow. This improves the algorithm of \cite{HenzingerRG00} when the vertex connectivity is high.

\section{Overview}
\label{sec:overview}

\subsection{Exact Algorithm}
\label{sub:overview exact}

To illustrate the main idea, let us sketch our algorithm with running time $\tilde{O}(m+n^{4/3})$ only on
an {\em undirected graph} with $m=O(n)$ and $k=O(1)$. 
This regime is already very interesting, because the
best bound has been $\tilde{O}(n^{2})$ for nearly 50 years \cite{Kleitman1969methods}.
Throughout this section,  $N(C)$ is a set of neighbors of nodes in $C\subseteq V$ that are not in $C$, and $E_G(S, T)$ is the set of edges between (not necessarily disjoint) vertex sets $S$ and $T$ in $G$ (the subscript is omitted when the context is clear).  A vertex partition $(A,S,B)$ is called a {\em separation triple} if $A, B\neq \emptyset$ and there is no edge between $A$ and $B$, i.e., $N(A)=S=N(B)$. 

Given a graph $G=(V,E)$ and a parameter $k$, our goal is to either return
a set $C\subset V$ where $|N(C)| < k$ or certify that $\kappa_G\ge k$. 
Our first step is to find a sparse subgraph $H$ of $G$ where $\kappa_H=\min\{\kappa_G,k\}$
using the algorithm by Nagamochi and Ibaraki \cite{NagamochiI92}.
The nice property of $H$ is that it is formed by a union of $k$
disjoint forests, i.e. $H$ has \emph{arboricity} $k$. In particular,
for any set of nodes $C$, we have $|E_{H}(C,C)|\le k|C|$. 
As the algorithm only takes linear time, from
now, we treat $H$ as our input graph $G$. 

The next step has three cases. First, suppose there is a separation triple $(A,S,B)$
where $|S| < k$ and $|A|,|B|\ge n^{2/3}$. Here, we sample $\tilde{O}(n^{1/3})$
many pairs $(x,y)$ of nodes uniformly at random. 
With high probability, one of these pairs is such that $x\in A$ and
$y\in B$. 
In this case, it is well known (e.g. \cite{Even75}) that one can modify the graph and run a max $xy$-flow algorithm. 
Thus, for each pair $(x,y)$, we run Ford-Fulkerson max-flow algorithm in time $O(k m)=O(n)$
to decide whether $\kappa(x,y)<k$ and if so, return the corresponding cut.
So w.h.p. the algorithm returns set $C$ where $|N(C)| < k$
in total time $\tilde{O}(n^{1+1/3})$.

The next case is when all separation triples $(A,S,B)$ where $|S| < k$
are such that either $|A|<n^{2/3}$ or $|B|<n^{2/3}$. Suppose w.l.o.g.
that $|A|<n^{2/3}$. By a binary search trick, we can assume to know the
size $|A|$ up to a factor of $2$. Here, we sample $\tilde{O}(n/|A|)$
many nodes uniformly at random. For each node $x$, we run the local vertex connectivity subroutine from \Cref{thm:local informal} where the parameter $k$ in \Cref{thm:local informal} is set to be $k-1$.
Note that the volume of $A$ is
$$\vol(A)=2|E(A,A)|+|E(A,S)|=O(k|A|)=O(|A|)$$ 
where the second equality is because
$G$ has arboricity $k$ and $|S| < k$ (also recall that we only consider $m=O(n)$ and $k=O(1)$ in this subsection). We set the parameter
$\nu=\Theta(|A|)$. With high probability, we have
that one of the samples $x$ must be inside $A$. Here, the local-max-flow
cannot be in the second case, and will return a set $C$ where $|N(C)| < k$,
which implies that $\kappa_G < k$. The total running time is
$\tilde{O}(n/|A|)\times\tilde{O}(|A|^{1.5})=\tilde{O}(n^{1+1/3})$
because $|A|<n^{2/3}$.

The last case is when $\kappa_G \ge k$. Here, both of Ford-Fulkerson
algorithm and local max flow algorithm will never return any set $C$
where $|N(C)|< k$. So we can correctly report that $\kappa_G \ge k$.
All of our techniques generalize to the case when $\kappa_G$ is not constant. 

%
%
%
%
%

\subsection{Local Vertex Connectivity}
\label{sub:overview flow}

In this section, we give a high-level idea how to obtain our local vertex
connectivity algorithm in \Cref{thm:local informal}. Recall from the
introduction that there are two tasks which are to design an \emph{augmented
	graph} and to devise a \emph{local flow-based algorithm} running on
such augmented graph. We have two goals: 1) the running time of our algorithm is \emph{local};
i.e., it does not depend on the size of the whole graph and
2) the local flow-based algorithm's output should be useful for our
application.

\paragraph{The local time principles.}
We first describe high-level principles on how to design the augmented
graph and the local flow-based algorithm so that the running time
is local\footnote{In fact, these are also principles behind all previous local flow-based
	algorithms. To the best of our knowledge, these general principles
	have not been stated. We hope that they explain previous seemingly
	ad-hoc results.}.\thatchaphol{From reviewer1: items (1) and (2).  I found this whole discussion confusing, and probably unnecessary.  In general, the authors should strive to make the paper as short as possible, while still keeping readability.}
\begin{enumerate}
	\item \emph{Augmented graph is absorbing}: Each node $u$ of the augmented
	graph is a \emph{sink} that can ``absorb'' flow proportional to
	its degree $\deg(u)$. More formally, each node $u$ is connected
	to a \emph{super-sink} $t$ with an edge $(u,t)$ of capacity $\alpha\deg(u)$
	for some constant $\alpha$. In our case, $\alpha=1$.
	\item \emph{Flow algorithm tries to absorb before forward}: Suppose that
	a node $u$ does not fully absorb the flow yet, i.e. $(u,t)$ is not saturated.
	When a flow is routed to $u$, the local flow-based algorithm must
	first send a flow from $u$ to $t$ so that the sink at $u$ is fully
	absorbed, before forwarding to other neighbors of $u$. Moreover,
	the absorbed flow at $u$ will stay at $u$ forever. 
\end{enumerate}
We give some intuition behind these principles. The second principle
resembles the following physical process. Imagine pouring water on
a compartment of an ice tray. There cannot be water flowing out of
an unsaturated compartment until that compartment is saturated. So
if the amount of initial water is small, the process will stop way before
the water reaches the whole ice tray. This explains in principle why
the algorithm needs not read the whole graph. 

The first principle allows us to argue why the cost of the algorithm
is proportional to the part of the graph that is read. Very roughly,
the total cost for forwarding the flow from a node $u$ to its neighbors
depends on $\deg(u)$, but at the same time we forward the flow only after it is already fully absorbed at $u$. 
This allows us to charge the total
cost to the total amount of absorbed flow, which in turn is small
if the initial amount of flow is small.

\paragraph{Augmented graph. }
Let us show how to design the augmented graph in the context of \emph{edge
	connectivity} in undirected graphs first. The construction is simpler
than the case of vertex connectivity, but already captures the main
idea. We then sketch how to extend this idea to vertex connectivity.

Let $G=(V,E)$ be an undirected graph with $m$ edges and $x\in V$
be a node.  Consider any numbers $\nu, k>0$ such that  
\begin{align}
2\nu k+\nu+1&\le2m. \label{eq:intro:localcut_condition}
\end{align}
We construct an undirected graph $G'$ as follows. The
node set of $G'$ is $V(G')=\{s\}\cup V\cup\{t\}$ where $s$ and
$t$ is a super-source and a super-sink respectively. For each node
$u$, add $(u,t)$ with capacity $\deg_{G}(u)$. (So, this satisfied
the first local time principle.)
For each edge $(u,v)\in E$, set
the capacity to be $2\nu$. Finally, add an edge $(s,x)$ with capacity
$2\nu k+\nu+1$.
\begin{theorem}\label{thm:intro:split_graph}
	Let $F^{*}$ be the value of the s-t
	max flow in $G'$. We have the following: 
	\begin{enumerate}
		\item If $F^{*}=2\nu k+\nu+1$, then there is no vertex partition $(S,T)$ in $G$ where
		$S\ni x$, $\vol(S)\le\nu$ and $|E(S,V-S)|\le k$.
		\item If $F^{*}\le2\nu k+\nu$, then there is a vertex partition $(S,T)$ in $G$ where
		$S\ni x$ and $|E(S,V-S)|\le k$. 
	\end{enumerate}
\end{theorem}
\begin{proof}
	To see (1), suppose for a contradiction that there is such a partition $(S,T)$
	where $S\ni x$. Let $(S',T')=(\{s\}\cup S,T\cup\{t\})$. The edges between $S'$ and $T'$ has total capacity
	\danupon{I don't fully see the first equality.}
	$$c(E_{G'}(S',T'))=2\nu|E_{G}(S,V-S)|+\vol_{G}(S)\le2\nu k+\nu.$$ 
	So $F^*\leq 2\nu k+\nu$, 
	a contradiction. To see (2), let $(S',T')=(\{s\}\cup S,T\cup\{t\})$
	be a min $st$-cut in $G'$ corresponding to the max flow, i.e. by the min-cut max-flow theorem, the	edges between $S'$ and $T'$ has total capacity

	\begin{align}\label{eq:intro:localcut_proof2}
	c(E_{G'}(S',T'))\leq 2\nu k+\nu. 
	\end{align}
	Observe that
	$S'\neq\{s\}$ and $S\ni x$ because the edge $(s,x)$ has capacity
	strictly more than $2\nu k+\nu$. Also, $T'\neq\{t\}$
	because edges between $\{s\}\cup V$ and $\{t\}$ has total capacity $\vol(V)=2m>2\nu k+\nu$ (the inequality is because of \Cref{eq:intro:localcut_condition}).
	So $(S,T)$ gives a cut in $G$ where $S\ni x$. Suppose that $|E_{G}(S,T)|\ge k+1$,
	then $c(E_{G'}(S',T'))\ge2\nu(k+1)=2\nu k+2\nu>2\nu k+\nu$ which
	contradicts \Cref{eq:intro:localcut_proof2}. 
\end{proof}

Observe that the above theorem is similar to \Cref{thm:local informal}
except that it is about edge connectivity. To extend this idea to
vertex connectivity, we use a standard transformation as used in \cite{EvenT75,HenzingerRG00}
by constructing a so-called \emph{split graph}. In our split graph,
for each node $v$, we create two nodes $v_{\text{in}}$ and $v_{\out}$.
For each edge $(u,v)$, we create an edge $(u_{\out},v_{\text{in}})$
with infinite capacity. There is an edge $(v_{\text{in}},v_{\text{out}})$
for each node $v$ as well. Observe that a cut set with finite capacity
in the split graph corresponds to a set of nodes in the original graph.
Then, we create the augmented graph of the split graph in a similar manner as above, e.g. by adding nodes $s$ and $t$ and an edge
$(s,x)$ with $2\nu k+\nu+1$. The important point is that we
set the capacity of each $(v_{\text{in}},v_{\text{out}})$ to be $2\nu$.
The proof of \Cref{thm:local informal} (except the statement about
the running time) is similar as above (see \Cref{sec:augmented graph}
for details).

\paragraph{Local flow-based algorithm.}
As discussed in the introduction, we can in fact adapt previous local flow-based algorithms 
to run on our augmented graph and they can decide the two cases in \Cref{thm:local informal} 
(i.e. whether there is a small vertex cut ``near'' a seed node $x$).
\Cref{thm:intro:split_graph} in fact already allows us to achieve this with slower running time than the desired $\tilde O(\nu^{1.5}k)$ by implementing existing local flow-based algorithms. 
For example, the algorithm by \cite{OrecchiaZ14}, which is a ``localized'' version of Goldberg-Rao algorithm \cite{GoldbergR98},
can give a slower running time of $\tilde O((\nu k)^{1.5})$. 
Other previous local flow-based algorithms that we are aware of (e.g. \cite{OrecchiaZ14,HenzingerRW17,WangFHMR17,VeldtGM16}) give even slower running time (even after appropriate adaptations).

We can speed up the time to $\tilde O(\nu^{1.5}k)$ by exploiting the fact that 
our augmented graph is created from a split graph sketched above.
This improvement resembles the idea by Hopcroft and Karp \cite{HopcroftK73} (see also \cite{EvenT75,Karzanov1974determining}) which yields
an $O(m\sqrt{n})$-time algorithm for computing $s$-$t$ {\em unweighted} vertex connectivity. 
The idea is to show that Dinic's algorithm with running time $O(m\min\{\sqrt{m},n^{2/3}\})$ on a general unit-capacity graph
can be sped up to $\tilde{O}(m\sqrt{n})$ when run on a special graph called ``unit network''.
It turns out that unit networks share some structures with our split graphs, allowing us to apply a similar idea.  
Although our improvement is based on a similar idea, it is more complicated to implement this idea on our split graph since it is weighted.

Finally, we ``localize'' our improved algorithm by enforcing the second local time principle.
Our way to localize the algorithm goes hand in hand with the way Orecchia
and Zhu \cite{OrecchiaZ14} did to the standard Goldberg-Rao algorithm
(see \Cref{sec:local binary blocking flow} for details).

\section{Preliminaries}\label{sec:prelim}

\subsection{Directed Graph}

Let $G = (V,E)$ be a \textit{directed} graph where $|V| = n$ and $|E|
= m$. We assume that $G$ is strongly-connected. Otherwise, we can list
all strongly connected components in linear time by a standard
textbook algorithm.  We also assume that $G$ is \textit{simple}. That is, $G$
does not have a duplicate edge. Otherwise, we can simplify the graph
in linear time by removing duplicate edges. 
For any edge $(u,v)$, we denote $e^{R} = (v,u)$. For any directed
graph $G = (V,E)$, the {\em reverse graph} $G^{R}$ is $G^R = (V, E^R)$ where
$E^R = \{ e^R \colon e \in E \}$.

\begin{definition}[$\delta$, $\deg$, $\vol$, $N$]
Definitions below are defined for any vertex $v$ on graph $G$ and subset of vertex $U \subseteq V$.

\begin{itemize}[noitemsep,nolistsep]
	\item $\delta_G^{\text{in}}(v) = \{(u, v)\in E\} $ and
          $\delta_G^{\text{in}}(U) = \{(x, y)\in E \colon x\notin U, y\in V\}$; i.e. they are the sets of edges entering $v$ and $U$ respectively.
	
	\item Analogously, $\delta_G^{\text{out}}(v)$ and $\delta_G^{\text{out}}(U)$ are the sets of edges leaving $v$ and $U$ respectively.
	\item $\deg_G^{\text{in}}(v) = |\delta_G^{\text{in}}(v)|$ and $\deg_G^{\text{out}}(v) = |\delta_G^{\text{out}}(v)|$; i.e. they are the numbers of edges entering and leaving $v$ respectively.
	\item 
	$\text{vol}_G^{\text{out}}(U) = \sum_{v \in U}
	\text{deg}_G^{\text{out}}(v)$ and $\text{vol}_G^{\text{in}}(U)= \sum_{v \in U}
	\text{deg}_G^{\text{in}}(v)$. Note that $\text{vol}_G^{\text{in}}(V) =
	\text{vol}_G^{\text{out}}(V) = m$. 
	\item $N_G^{\text{in}}(v) = \{u \colon (u,v)\in E\}$ and $N_G^{\text{out}}(v) = \{u \colon (v, u)\in E\}$; i.e. they are sets of in- and out-neighbors of $v$, respectively. 
	\item  $N_G^{\text{in}}(U) = \bigcup_{v \in U}
          N_G^{\text{in}}(v) \setminus U $ and $N_G^{\text{out}}(U) =
          \bigcup_{v \in U} N_G^{\text{out}}(v) \setminus U$. Call
          these sets  {\em external in-neighborhood of $U$} and  {\em
            external out-neighborhood of $U$}, respectively. \qedhere  
\end{itemize}
\end{definition}

\begin{definition}[Subgraphs]
For a set of vertices $U \subseteq
V$, we denote $G[U]$ as a subgraph of $G$ induced by $U$. Denote for
any vertex $v$, any subset of vertices $U \subseteq V$, any edge $e
\in E$, and any subset of edges $F \subseteq E$, 
\begin{itemize}[noitemsep,nolistsep]
\item $G \setminus v
= (V \setminus \{ v\}, E)$, 
\item $G \setminus U = (V \setminus U, E),$ 
\item $G
\setminus e = (V, E\setminus \{ e \}), $ and 
\item $ G \setminus F = (V,
E\setminus F)$. 
\end{itemize}
We say that these graphs \textit{arise} from $G$ by deleting
$v, U, e,$ and $F$, respectively. 
\end{definition}

\begin{definition}[Paths and reachability]
	For $s,t \in V$, we say a path $P$ is an {\em $(s,t)$-path}
	if $P$ is a directed path
	starting from $s$ and ending at $t$. 
	For any $S, T \subseteq V$, we say
	$P$ is an {\em $(S,T)$-path} if $P$ starts with some vertex in $S$ and ends
	at some vertex in $T$.  
	We say that a vertex $t$ is \textit{reachable} from
	a vertex $s$ if there exists a $(s,t)$-path $P$.
	Moreover, if a node $v$ is in such path $P$, then  we say that $t$ is reachable from $s$ {\em via $v$}.
\end{definition}

\begin{definition}[Edge- and
  Vertex-cuts] \label{def:edge-or-vertex-cuts}Let $s$ and $t$ be any
  distinct vertices. Let $S, T \subset V$ be any disjoint non-empty subsets of vertices. We call any subset of edges $C \subseteq E$ (respectively any subset of vertices $U\subseteq V$):
\begin{itemize}[noitemsep,nolistsep]

\item an {\em $(S,T)$-edge-cut} (respectively an {\em
    $(S,T)$-vertex-cut} ) if there is no $(S,T)$-path in $G\setminus
  C$ (respectively if there is no $(S,T)$-path in $G\setminus U$ {\em
    and} $S \cap U = \emptyset, T \cap U = \emptyset$),
\item an {\em $(s,t)$-edge-cut} (respectively an {\em $(s,t)$-vertex-cut} ) if there is no $(s,t)$-path in $G\setminus C$ (respectively if there is no $(s,t)$-path in $G\setminus U$ {\em and} $s, t\notin U$),
\item an  {\em $s$-edge-cut} (respectively {\em $s$-vertex-cut}) if it is an $(s,t)$-edge-cut (respectively $(s,t)$-vertex-cut) for some vertex $t$, and
\item an  {\em edge-cut} (respectively {\em vertex-cut}) if it is an $(s,t)$-edge-cut (respectively $(s,t)$-vertex-cut) for some distinct vertices $s$ and $t$. In other words,  $G\setminus C$ (respectively $G\setminus U$) is not strongly connected. 
\end{itemize}	
If the 
graph has capacity function $c : E \rightarrow
\mathbb{R}_{\geq 0}$ on edges, then $c(C)= \sum_{e \in C} c_e$ is the total
capacity of the cut $C$.
\end{definition}

\begin{definition} [Edge set]
 We define $E(S,T)$ as the set of edges $\{ (u,v) \colon u \in  S, v \in T\}$. 
\end{definition}

\begin{definition}[Vertex partition]\label{def:separation_triple}
Let $S, T \subset V$.  We say that $(S,T)$ is a \textit{vertex
  partition} if $S$ and $T$ are not empty, and $S \sqcup T = V$. In
particular, $E(S,T)$ is an $(x,y)$-edge-cut for some $x \in S, y \in
T$. 
\end{definition}

\begin{definition}[Separation triple]\label{def:separation_triple}
We call $(L,S,R)$ a \textit{separation triple} if $L, S,$ and $R$
partition the vertex $V$ in $G$ where $L$ and $R$ are non-empty, and
there is no edge from $L$ to $R$. 
\end{definition}

Note that, from the above definition, $S$ is an $(x,y)$-vertex-cut for
any $x \in L$ and $y \in R$. 

\begin{definition}[Shore] \label{def:shore}
We call a set of vertices $S\subseteq V$ an {\em out-vertex shore} (respectively  {\em in-vertex shore}) if
$N_G^{\text{out}}(S)$ (respectively $N_G^{\text{in}}(S)$) is a
vertex-cut.%
\end{definition} 

\begin{definition}[Vertex connectivity $\kappa$]\label{def:VertexConnectivity}
We define vertex connectivity $\kappa_G$ as the 
minimum cardinality vertex-cut or $n-1$ if no vertex cut exists. More
precisely,  for distinct $x,y \in V$, define $\kappa_G(x,y)$ as the smallest cardinality of $(x,y)$-vertex-cut if
exists. Otherwise,  we define $\kappa_G(x,y) = n-1$. Then, $\kappa _G=
\min \{ \kappa_G (x,y) \text{ } | \text{ } x,y \in V, x\neq y \}$. We
drop the subscript when $G$ is clear from the context. 
\end{definition}

Let $d^{\out}_{\text{min}} = \min_v \text{deg}_G^{\text{out}}(v)$ and let $v_{\text{min}}$ be
any vertex whose out-degree is $d^{\out}_{\text{min}}$. If $d^{\out}_{\text{min}} = n - 1$, then $G$ is complete, meaning that $\kappa_G =
n-1$. Otherwise, $\delta_G^{\text{out}}(v_{\text{min}})$ is a
vertex-cut. Hence, $\kappa_G \leq |\delta_G^{\text{out}}(v_{\text{min}})| =
d^{\out}_{\text{min}}$.  So, we have the following observation. 

\begin{observation} \label{obs:kappa-degree}
$\kappa_G \leq  d^{\out}_{\min}$ 
\end{observation}

\begin{proposition} [\cite{HenzingerRG00}] 
There exists an algorithm that takes as input graph $G$, and two vertices $x,y
\in V$ and an integer $k > 0$ and in $\ot(\min( km )$ time outputs either an
out-vertex shore $S$ containing $x$ with $|N_G^{\text{out}}(S)| =
\kappa_G(x,y) \leq k$ and $y $ is in the corresponding in-vertex shore, or an ``$\perp$''
symbol indicating  that no such shore exists and thus $\kappa_G(x,y) > k$.  
\label{pro:easyff}
\end{proposition}

\subsection{Undirected Graph} 

Let $G = (V,E)$ be an undirected graph. We assume that $G$ is simple,
and connected. 

\begin{theorem} [\cite{NagamochiI92}] \label{thm:sparsification}
There exists an algorithm that takes as input undirected graph $G =
(V,E)$, and in $O(m)$ time outputs a sequence of forests $F_1, F_2,
\ldots, F_n$ such that  each forest subgraph $H_k= (V,
\bigcup_{i=1}^kF_i)$ is $k$-connected if $G$ is $k$-connected. $H_k$
has aboricity $k$.  For any set of vertices $S$, we have $E_{H_k}(S,S)
\leq k|S|$. In particular, the number of edges in $H_k$ is at most $kn$.

\end{theorem}

To compute vertex connectivity in an undirected graph, we turn it into
a directed graph by adding edges in forward and backward directions
and run the directed vertex connectivity algorithm. 
  
	\section{Local Vertex Connectivity}
\label{sec:local_flow}

Recall that a directed graph $G = (V,E)$ is strongly connected where $|V| = n$ and $|E| = m$. 
 \begin{theorem}
There is an algorithm that takes as input a pointer to any vertex $x \in
V$ in
an adjacency list representing a strongly-connected directed graph $G=(V, E)$,
positive integer  $\nu$ (``target volume''), positive integer $k$ (``target
$x$-vertex-cut size''), and  positive real $\epsilon$  satisfying 
\begin{align}\label{eq:local_condition}
  \nu/  \epsilon + \nu < m, \quad  (1+\epsilon)( \frac{2\nu}{\epsilon k}+ k)
  < n \quad  \mbox{and}\quad  \dmin \geq k 
\end{align}
or, 
\begin{align}\label{eq:local_condition_dense}
  \nu/\epsilon  + (1+\epsilon) nk < m, \quad  \mbox{and}\quad  \dmin \geq
  k \end{align}

 and in $\ot( \frac{\nu^{3/2}}{\epsilon^{3/2} k^{1/2}})$ time outputs either   
	\begin{itemize}[noitemsep,nolistsep]
		\item a vertex-cut $S$ corresponding to the separation triple $(L,S,R), x \in L$ such that   

		\begin{align}\label{eq:found_cut}
                  |S| \leq (1+\epsilon) k \quad \mbox{and}\quad
                  \vol_G^{\text{out}}(L) \leq \nu/\epsilon + \nu +1 , \mbox{or}
		\end{align}
		\item the ``$\perp$'' symbol indicating that there is no
                  separation triple $(L,S,R), x \in L$ such that  
	
		\begin{align}\label{eq:unfound_cut}
		|S| \leq k \quad\mbox{and}\quad \vol_G^{\text{out}}(L) \leq \nu.
		\end{align}
	\end{itemize}
	\label{thm:local-vertex-connectivity} 
\end{theorem} 

By setting $\epsilon = 1/(2k)$, we get the exact version for the size
of vertex-cut. Observe that \Cref{eq:found_cut} is changed to $|S|
\leq (1+ 1/(2k))k = k + 1/2$. So $|S| \leq k$ since $|S|$ and $k$ are
 integers. 

\begin{corollary} 	\label{cor:exact-local-vertex-connectivity} 
There is an algorithm that takes as input a pointer to any vertex $x \in
V$ in
an adjacency list representing a strongly-connected directed graph $G=(V, E)$,
positive integer  $\nu$ (``target volume''), and  positive integer $k$(``target
$x$-vertex-cut size'') satisfying
\Cref{eq:local_condition}, or \Cref{eq:local_condition_dense} where
$\epsilon = 1/(2k)$, 
 and in $\ot( \nu^{3/2}k )$ time outputs either  
	\begin{itemize}[noitemsep,nolistsep]
		\item a vertex cut $S$ corresponding to the separation triple $(L,S,R), x \in L$ such that   
		\begin{align}\label{eq:exact_found_cut}
                  |S| \leq  k \quad \mbox{and}\quad
                  \vol_G^{\text{out}}(L) \leq 2\nu k + \nu +1 , \mbox{or}
		\end{align}
		\item the ``$\perp$'' symbol indicating that there is no
                  separation triple $(L,S,R), x \in L$ such that  
		
		\begin{align}\label{eq:exact_unfound_cut}
		|S| \leq k \quad\mbox{and}\quad \vol_G^{\text{out}}(L) \leq \nu.
		\end{align}
	\end{itemize}

\end{corollary}

The rest of this section is devoted to proving the above theorem. For
the rest of this section, fix $x$, $\nu$, $k$  and $\epsilon$ as in the theorem statement. 

\subsection{Augmented Graph and Properties}
\label{sec:augmented graph}
\begin{definition}  [Augmented Graph $G'$] \label{def:aug_graph}
Given a directed uncapacitated graph $G= (V,E)$, we define a directed capacitated graph $(G',c_{G'}) =
((V',E'), c_{G'})$  where
\begin{align}\label{eq:aug_graph}
V' = V_{\text{in}} \sqcup V_{\text{out}} \sqcup \{ s,t\}  \quad\mbox{and}\quad E' = E_\nu \sqcup E_\infty \sqcup
E_{\text{deg}} \sqcup \{ (s,x_{\out}) \},
\end{align}
where $\sqcup$	denotes disjoint union of sets, $s$ and $t$ are additional vertices not in $G$, and sets in \Cref{eq:aug_graph} are defined as follows. 

\begin{itemize}[noitemsep,nolistsep]
\item For each vertex  $v \in V \setminus \{ x \} $, we create vertex
  $v_{\text{in}}$ in set $V_{\text{in}}$  and $v_{\text{out}}$ in set $V_{\text{out}} $. For the vertex $x$,  we add only $x_{\out}$ to $V_{\out}$.
\item $E_{\nu} = \{ (v_{\text{in}}, v_{\text{out}}) \colon  v \in V \setminus \{ x \}\}$.   
\item $E_{\infty} = \{ (v_{\text{out}}, w_{\text{in}}) \colon    (v,w) \in E \} $. 
\item $E_{\text{deg}} = \{ (v_{\text{out}},t) \colon  v \in V_{\text{out}} \} $. 
\end{itemize}

Finally, we define the capacity function $c_{G'} : E' \rightarrow
\mathbb{R}_{\geq 0} \cup \{ \infty \}$ as:

\begin{equation*}
c_{G'}(e)  = \begin{cases}
\nu/ (\epsilon k)          & \text{if } e = (v_{\text{in}},      v_{\text{out}}) \in E_\nu \\
\text{deg}_G^{\text{out}}(v)               & \text{if }
e = (v_{\text{out}}, t)  \in E_{\text{deg}} \\ 
\nu/\epsilon + \nu + 1 & \text{if } e = (s, x_{\out}) \\
\infty & \text{otherwise}  
\end{cases}
\end{equation*}
\end{definition}

\begin{lemma}\label{lem:aug_graph_properties}
	Let $C^*$ be the minimum-capacity $(s,t)$-cut in $G'$.  Recall
        that  $c_{G'}(C^*)$ is its capacity and $\nu$ and $k$ satisfy
        \Cref{eq:local_condition} or \Cref{eq:local_condition_dense} .
\begin{enumerate}[noitemsep,nolistsep,label=(\Roman*)]
	\item \label{item:aug_graph_properties_one} If there exists a
          separation triple $(L,S,R), x \in L$ in $G$ 

	satisfying \Cref{eq:unfound_cut},
	then $c_{G'}(C^*) \leq \nu/\epsilon + \nu$.  
	\item \label{item:aug_graph_properties_two} If $c_{G'}(C^*)
          \leq \nu / \epsilon + \nu$, then there exists a separation
          triple $(L,S,R), x \in L$ in $G$ satisfying \Cref{eq:found_cut}. 

\end{enumerate}

\end{lemma}

We prove \Cref{lem:aug_graph_properties} in the rest of this
subsection. %

We define useful notations. For $U \subseteq V $ in $G$,
define $V_{\text{out}}(U) =\{ v_{\text{out}} \text{ } | \text { } v
\in U\} \subseteq V_{\text{out}}$ in $G'$. Similarly, we define
$V_{\text{in}}(U) = \{ v_{\text{in}} \text{ } | \text { } v
\in U\} \subseteq V_{\text{in}}$ in $G'$ .

We first introduce a standard \textit{split graph} $SG$ from $G'$. 
\begin{definition} [Split graph $SG$]\label{def:sg}
Given $G'$, a \textit{split graph} $SG$ is an induced graph $SG = G'[W]$ where
$$W = V_{\text{in}} \sqcup V_{\text{out}} \sqcup \{ x \}, $$
with capacity function $c_G'(e)$ restricted to edges in $G'[W]$ where
the edge set of $G'[W]$ is $E_\nu \sqcup  E_\infty$.  
\end{definition}

\begin{proof}[Proof of
  \Cref{lem:aug_graph_properties}\ref{item:aug_graph_properties_one}]
We fix a separation triple $(L,S,R)$ given in the statement.  Since $x \in L$, $S$ is an
$(x,y)$-vertex-cut for some $y \in R$ by \Cref{def:separation_triple}.  

Let $C = \{ (u_{\text{in}}, u_{\text{out}}) \colon u \in S \}$. It is easy to see that $C$ is an $(x_{\out},y_{\textin})$-edge-cut in the split graph
$SG$. Since $S$ is an $(x,y)$-vertex-cut, there is no vertex-disjoint paths from $x$
to $y$ in $G \setminus S$. By transforming from $G$ to $G'$, vertex
$y$ in $G$ becomes $y_{\textin}$ and $y_{\out}$ in $G'$. Since $S$ separates $x$ and $y$ in $G$, by
construction of $C$, $C$ must separate $x$ and $y_{\textin}$ in $G'$ and thus in $SG$. Therefore
there is no $(x_{\out},y_{\textin})$-path  in $SG \setminus C$, and
the claim follows. 

In $G'$, we define an edge-set $C' = C \sqcup  \{ (v,t)  \text{ } |
\text{ } v \in V_{\text{out}}(L) \}$.  It is easy to see that $C'$ is an $(s,t)$-edge-cut in $G'$.   Since $C$ is an
$(x_{\out},y_{\textin})$-edge-cut in the split graph $SG$.  The graph $ G'
\setminus C$ has no $(s,V_{\text{out}} ( S \sqcup R))$-paths. 
Since $G$ is strongly connected, the sink vertex
$t$ in $G' \setminus C$ is reachable from $s$ via only vertices in
$V_{\text{out}}(L)$. Hence, it is enough to remove the
edge-set $\{ (v,t)  \colon v \in V_{\text{out}}(L) \}$ to
disconnect all $(s,t)$-paths in $G' \setminus C$, and the claim
follows.  
 
We now compute the capacity of the cut $C'$. 
\begin{align*} 
c_{G'}(C') &= c_{G'}(  C \sqcup  \{ (v,t) \colon v \in V_{\text{out}} (L)  \} ) \\
& = c_{G'}( C ) + c_{G'}(  \{ (v,t)  \colon v \in V_{\text{out}} (L) \}) \\
&= \nu |S|/ (\epsilon k) + \sum_{v \in S}
  \text{deg}_G^{\text{out}}(v) \\
& =  \nu |S| / (\epsilon k) + \text{vol}_G^{\text{out}}(S) \\
&  \leq  \nu/ \epsilon + \nu
\end{align*}
The last inequality follows from $|S| \leq k$ and $\text{vol}_G^{\text{out}}(S) \leq \nu$. 

Hence, the capacity of the minimum $(s,t)$-cut $C^*$ is $ c_{G'}(C^*) \leq c_{G'}(C')
\leq \nu/\epsilon + \nu  $. 
\end{proof}

Before proving \Cref{lem:aug_graph_properties}\ref{item:aug_graph_properties_two}, we
 observe structural properties of an $(s,t)$-edge-cut in $G'$.

\begin{definition}
Let $\cC$ be the set of $(s,t)$-cuts of finite capacities in $G'$. We define three
subsets of $\cC$ as,
\begin{itemize}[noitemsep,nolistsep]
\item $\cC_1 =  \{ C \colon C \in \cC $, and  one side of vertices in $ G'\setminus C $ contains
  $ s $ or $ t $ as a singleton $ \} $.
\item $\cC_2 = \{ C \colon C \in \cC \setminus \cC_1$,
  and $ C $ is an $ (\{ s \}
  \sqcup V_{\text{in}}, \{ t \})$-edge-cut$  \} $.
\item $\cC_3 = \{ C \colon C \in \cC \setminus \cC_1 $,
  and $ C $ is an $
  (\{ s \}, \{v_{\text{in}},t \})\text{-edge-cut } $ for some $
  v_{\text{in}} \in V_{\text{in}} \} $. 
\end{itemize}
\label{def:cut-structure}
\end{definition}
Observe that three partitions in \Cref{def:cut-structure}
formed a complete set $\cC$ and are pairwise disjoint by
\Cref{def:edge-or-vertex-cuts}, and by the  construction of $G'$. 

\begin{observation} 
\label{obs:3partitions}
$$ \cC = \cC_1 \sqcup \cC_2 \sqcup \cC_3 $$
\end{observation}

\begin{proposition} We have the following lower bounds on cut capacity
  for cuts in $\cC_1 \sqcup \cC_2$. 
\begin{itemize}[noitemsep,nolistsep]
\item For all $C \in \cC_1$, $c_{G'}(C) \geq \min ( \nu/\epsilon + \nu
  +1 , m )$
\item For all $C \in \cC_2$, $c_{G'}(C) \geq \min ( \nu/\epsilon + \nu
  +1,  \max((n-(1+\epsilon)k)k, m - (1+\epsilon)nk))$ 
\end{itemize}
\label{pro:cut-structure}
\end{proposition}
\begin{proof}

By \Cref{def:cut-structure}, any $C \in \cC_1$ contains $(s,x_{\out})$ or
$E_{\text{deg}}$. So, $C$ has capacity $c_{G'}(C) \geq
\min(\nu/\epsilon + \nu +1, \sum_{v \in V} \text{deg}_G^{\text{out}}(v))=
\min(\nu/\epsilon  + \nu +1, m)$.  

Next,  we show that if $C \in \cC_2$, then $c_{G'}(C) \geq \min (
\nu/\epsilon + \nu + 1, \max((n-(1+\epsilon)k)k/2, m - (1+\epsilon)nk))$. By \Cref{def:cut-structure}, $C$ has finite capacity. We can write $C = E_\nu^* \sqcup E^*_{\text{deg}}$  where $E_\nu^*
\subseteq E_\nu$ and $E^*_{\text{deg}} \subseteq
E_{\text{deg}} $. 
 If $|E^*_\nu| > (1+\epsilon)k$, then, by construction of $G'$,
 $c_{G'}(C) \geq  \frac{\nu}{\epsilon k} |E_\nu^*| > \nu/\epsilon +
 \nu$. 

From now, we assume that $|E^*_\nu| \leq (1+\epsilon)k$. We show two
inequalities:

\begin{align} \label{eq:ineqPart1}
c_{G'}(C) \geq (n-(1+\epsilon)k)k
\end{align}

and

\begin{align} \label{eq:ineqPart2}
c_{G'}(C) \geq m - nk(1+\epsilon).
\end{align}

We claim that $|E^*_{\text{deg}}| \geq n-(1+\epsilon)k$.  Consider $G' \setminus C$. Let $S = \{ x \} \sqcup V_{\text{in}}$.  Observe that any $w
\in S$ cannot reach $t$ in $G' \setminus C$ since $C$ is an $ (\{ s \}
\sqcup V_{\text{in}}, \{ t \})$-edge-cut. So, for all $v_{\text{in}}
\in V_{\text{in}}$, we have $(v_{\text{in}},  v_{\text{out}}) \in C$ or
$(v_{\text{out}}, t) \in C$.  Since
$|E^*_\nu| \leq (1+\epsilon)k$, this means we can include edges of type
$(v_{\textin}, v_{\out})$ at most $(1+\epsilon)k$ edges. Hence,  the rest of the
edges must be of the form $(v,t)$, and thus $|E^*_{\text{deg}}| \geq
n- (1+\epsilon)k$. 

We now show \Cref{eq:ineqPart1}. 
 By \Cref{eq:local_condition} or \Cref{eq:local_condition_dense}, $\dmin \geq k$.
Since $C = E^*_\nu \sqcup E^*_{\text{deg}}$,  we have $c_{G'}(C) \geq
c_{G'}(E^*_{\text{deg}}) \geq |E^*_{\text{deg}}| \dmin \geq (n-(1+\epsilon)k)k.$

Finally, we show \Cref{eq:ineqPart2}.  Since  $|E^*_{\text{deg}}| \geq
n-(1+\epsilon)k$, $|E_{\text{deg}} \setminus E^*_{\text{deg}} | \leq (1+\epsilon)k$, and thus
$c_{G'}(E_{\text{deg}} \setminus E^*_{\text{deg}}) \leq (1+\epsilon)nk$, (recall
each vertex has degree at most $n-1$). Therefore, 

$$c_{G'}(C) \geq c_{G'}(E^*_{\text{deg}}) = \sum_{v} \text{deg}_{G}^{\out}(v) -
c_{G'}(E_{\text{deg}} \setminus E^*_{\text{deg}}) \geq m - (1+\epsilon)nk $$

\end{proof}

\begin{corollary} \label{cor:goodcut}
For all $C \in \cC$, if $c_{G'}(C) \leq \nu/\epsilon + \nu$, then $C \in \cC_3$ 
\end{corollary}
\begin{proof}
By \Cref{obs:3partitions} and  \Cref{pro:cut-structure} , it is enough
to show that $C \not \in \cC_1$ and $C \not \in \cC_2$ using \Cref{eq:local_condition}, or
\Cref{eq:local_condition_dense}. By either \Cref{eq:local_condition} or
\Cref{eq:local_condition_dense},  $ \nu/\epsilon + \nu < m$, and thus  $C \not
\in \cC_1$. Next, we show that $C \not \in \cC_2$. It is enough to
show that $\nu/\epsilon + \nu $ is smaller than one of two terms in
max.  If \Cref{eq:local_condition} is satisfied,  then $(1+\epsilon)(
2\nu/ (\epsilon k) + k) < n $. This implies  $\nu/\epsilon + \nu < (n-
(1+\epsilon)k) k.$%
If \Cref{eq:local_condition_dense} is satisfied, then  we immediately
get $\nu/\epsilon + \nu  < m - (1+\epsilon) nk $.
\end{proof}

We now ready to prove
\Cref{lem:aug_graph_properties}\ref{item:aug_graph_properties_two}. 
 
\begin{proof}[Proof of \Cref{lem:aug_graph_properties}\ref{item:aug_graph_properties_two}]\danupon{I
    didn't read this proof yet} 
In $G$, we show the existence of a separation triple $(L,S,R)$ where $x
\in L, |S| \leq (1+\epsilon)k$. 

The minimum $(s,t)$-cut in $G'$, $C^* $, is an $
  (s, \{v_{\text{in}},t \})\text{-edge-cut } $ (with finite capacity) for some $
  v_{\text{in}} \in V_{\text{in}} $. Since $c_{G'}(C^*) \leq 
  \nu/\epsilon + \nu $, by \Cref{cor:goodcut}, $C^* \in \cC_3$. 

We can write  $C^*= E^*_{\text{deg}} \sqcup E^*_{\nu}$ where $\emptyset \not =  E^*_{\text{deg}}
\subsetneq E_{\text{deg}} $ and $\emptyset \not = E^*_{\nu} \subsetneq E_{\nu}$ in
$G'$.  To see that $ E^*_{\nu} \not = \emptyset $, suppose otherwise, then $ C^*
$ must be in  $\cC_1$, a contradiction. 
 
It is easy to see that  $E^*_\nu$ is an $(x_{\out},v_{\text{in}})$-edge-cut in
$SG$. First of all, $E^*_\nu$ is the subset of edges in $SG$ by
\Cref{def:sg}.  Since $v_{\text{in}}$ is not reachable from $s$ in $G' \setminus C^*$
and $(s,x) \not \in C^*$, $x$ cannot reach $v_{\text{in}}$ in $G'
\setminus C^*$. Observe that  edges in $E_{\text{deg}}$ (and in
particular, $E^*_{\text{deg}}$)  have
no effect for reachability of the $(x_{\out},v_{\text{in}})$
path in $G'$.  Since $C^* =  E^*_{\text{deg}} \sqcup E^*_{\nu}$, only
edges in $E^*_\nu$ can affect the reachability of the
$(x,v_{\text{in}})$ path in $G'$.  So, when restricting $G'$ to $SG$,  $x_{\out}$ cannot reach
$v_{\text{in}}$ in $SG \setminus E^*_\nu$. Therefore, $E^*_\nu$ is an $(x_{\out},v_{\text{in}})$-edge-cut in
$SG$, and the claim follows.

To show a separation triple $(L,S,R)$, it is enough to define $S$, and show
that $S$ is an $(x,y)$-vertex-cut where $x \in L$ and $y \in R$. This
is because $L$ and $R$ can be found trivially when we remove $S$
from $G$. 

Let $S=  \{  u \in V \colon (u_{\text{in}},u_{\text{out}}) \in E^*_{\nu} \}$.  It is easy to see that $S$ is an $(x,y)$-vertex-cut in $G$ for some $y \in
V$. Since $E^*_{\nu} \not = \emptyset$ is an
$(x,y_{\textin})$-edge-cut in $SG$, $y_{\textin}$ is not reachable by
$x$  in $SG \setminus
E^*_{\nu}$. By construction of $G'$ (\Cref{def:aug_graph}), the corresponding out-vertex pair of $y_{\textin}$, $y_{\out}$, has
in-degree one from $y_{\textin}$. So, $y_{\out}$ in $SG \setminus
E^*_{\nu}$ is also not reachble by $x$. Hence, in $SG \setminus
E^*_{\nu}$, both $y_{\textin}$ and $y_{\out}$ are not reachable from
$x$. So, by the construction of $G'$,  and in $G \setminus U$, $y$ is
not reachable from $x$. Therefore, $U$ is an $(x,y)$-vertex-cut in $G$. 

Next, $|S| \leq (1+\epsilon)k$ since otherwise $c_{G'}(C^*)  >
(1+\epsilon)k (\nu/(\epsilon k)) = \nu/\epsilon+\nu $, a contradiction to the capacity of $C^*$.  
 

We next show that $\vol_G^{\text{out}}(L) \leq \nu/\epsilon + \nu +
1$. 

Let $E_{\text{deg}}(L) = \{ (v_{\text{out}},t) \colon v \in L \} $.  We claim that $ E^*_{\text{deg}} =
E_{\text{deg}}(L)$. Since $E^*_{\nu}$ is an  $(x_{\out},v_{\text{in}})$-cut in $SG$,   the graph  $G' \setminus
E^*_\nu$ has no $(s,V_{\text{out}} (S \sqcup R))$-paths. Since $G$ is strongly
connected, the sink vertex $t$ in $G' \setminus  E^*_\nu $ is
reachable from $s$ via only vertices in $V_{\text{out}}(L)$. Since $C^*$ is the minimum $(s,t)$-cut in $G'$, $E^*_{\text{deg}}$ only contains edges in $E_{\text{deg}}(L)$.  The claim follows. 
 
We now show that $\vol_G^{\text{out}}(L) \leq  \nu/\epsilon + \nu +1
$. By the previous claim, $c_{G'}(E^*_{\text{deg}}) = \sum_{v \in S}
\text{deg}_G^{\text{out}}(v) = \text{vol}_G^{\text{out}}(S).$  Also,
denote $F^*$ as the value of the maximum $(s,t)$-flow in $G'$. By
strong duality (max-flow min-cut theorem),  $C_{G'}(C^*) = F^*$. Note
that $F^* \leq \nu/\epsilon + \nu + 1$ since this corresponds to an
$(s,t)$-edge cut that contains edge $(s,x_{\out})$. Hence, 
\begin{align*}
\text{vol}_G^{\text{out}}(S) + c_{G'}(E^*_\nu) &=
c_{G'}(E^*_{\text{deg}})+ c_{G'}(E^*_\nu) \\
&= c_{G'}( E^*_{\text{deg}} \sqcup E^*_\nu) \\
&= c_{G'}(C^*) \\
& = F^*  \\ 
& \leq \nu/\epsilon + \nu +1.  
\end{align*}
Therefore, $\text{vol}_G^{\text{out}}(S) + c_{G'}(E^*_\nu) \leq
\nu/\epsilon + \nu +1$, and thus $\text{vol}_G^{\text{out}}(S)  \leq
\nu/\epsilon + \nu +1$ as desired.
\end{proof}

\subsection{Preliminaries for Flow Network and Binary Blocking Flow}

We define notations related flows on a capacitated directed graph $G = (V,E,c)$. We fix vertices $s$ as source and $t$ as sink.  

\begin{definition} [Flow] 
For a capacitated graph $G=(V,E,c)$, a \textit{flow} $f$ is a function
$f : E \rightarrow \mathbb{R}$ satisfying two conditions:
\begin{itemize} [noitemsep, nolistsep]
\item For any $(v,w) \in E, f(v,w) \leq c(v,w)$, i.e., the flow on
  each edge does not exceed its capacity.
\item For any vertex $v \in V \setminus \{s, t\},  \sum_{u: (u,v) \in E}
  f(u,v) = \sum_{w: (v,w) \in E} f(v,w)$, i.e., for each vertex except
  for $s$ or $t$,  the amount of incoming flow is equal to the amount of outgoing flow.
\end{itemize}
We denote $|f| = \sum_{v : (v,t) \in E} f(v,t)$ as the value of flow $f$. 
\end{definition}

\begin{definition}[Residual graph]
Given a capacitated graph $G = (V,E,c)$ and a flow function $f$, we define the
\textit{residual graph with respect to }  $f$ as $(G,c,f) = (V, E_f, c_f)$
where $E_f$ contains all edges $(v,w) \in E$ with $c(v,w) - f(v,w) >
0$.  Note that  $f(v,w)$ can be negative if the actual flow goes from $w$
to $v$, and $E_f$ may contain reverse edge $e^R$ to the original graph
$G$. We call an edge in $E_f$ as \textit{residual edge} with
\textit{residual capacity} $c_f(v,w) = c(v,w) - f(v,w)$. An edge in
$E$ is not in $E_f$ when the amount of flow through this edge equals
its capacity. Such an edge is called an \textit{saturated edge}. We
sometimes use notation $G_f$ as the shorthand for the residual graph
$(G,c,f)$ when the context is clear. 
\end{definition}

\begin{definition}[Blocking flow] \label{def:blockingflow}
Given a capacitated graph $G = (V,E,c)$, a \textit{blocking flow} is a
flow that saturates at least one edge on every $(s,t)$-path in $G$. 
\end{definition}

We will use \Cref{def:blockingflow} mostly on the residual graph $G_f$. 

Given a binary length function $\ell$ on $(G,c,f)$, we define
 a natural distance function to each vertex in $(G,c,f)$ under
 $\ell$. 

\begin{definition} [Distance function] \label{def:distance-function}
Given a residual graph $G_f$, and binary length function $\ell$, a
function $d_{\ell} : V \rightarrow \mathbb{Z}_{\geq 0}$ is a distance
function if $d(v)$ is the length of the shortest $(s,v)$-path in $G_f$
under the binary length function $\ell$
\end{definition}

For any $(v,w) \in E_f$, $d_{\ell}(v) +\ell(v,w) \geq d_{\ell}(w)$ by
\Cref{def:distance-function}. If $d_{\ell}(v) + \ell(v,w)=  d_{\ell}(w)$, then
we call $(v,w)$ \textit{admissible edge} under length function
$\ell$. 
 
We denote $E_a$ to be the set of admissible edges of $E_f$ in
$(G,c,f)$ under length function $\ell$.  

\begin{definition}[Admissible graph]
Given a residual graph $(G,c,f)$, and a length function function $\ell$, we
 define  an \textit{admissible graph} $A(G,c,f,\ell) = (G[E_a],c,f)$ to
 be an induced subgraph of $(G,c,f)$ that contains only admissible
 edges under length function $\ell$. 
\end{definition}

\begin{definition}[$\Delta'$-or-blocking flow]
For any $\Delta' > 0$, a flow is called a $\Delta'$-\textit{or-blocking flow} if it is a flow of
value exactly $\Delta'$, or a blocking flow. 
\end{definition}

\begin{definition} [Binary length function $\tilde \ell$] \label{def:gr-binary-length}
Given $\Delta > 0$, a capacitated graph $(G,c)$ and a flow $f$, we
define binary length functions $\hat \ell$ and $\tilde \ell$ for any edge $(u,v)$ in a residual graph $(G,c,f)$ as follows. 
$$
\hat \ell(u,v)  = \left\{ \begin{array}{rl}
  0 &\mbox{ if residual capacity $c(u,v) - f(u,v) \geq \Delta$} \\
  1 &\mbox{ otherwise}
       \end{array} \right.
$$

Let $\hat d(v)$ be the shortest path distance between $s$ and $v$ under the
length function $\hat \ell$. We define \textit{special edge} $(u,v)$ to be
an edge $(u,v)$ such that $\hat d(u) = \hat d(v), \Delta/2 \leq c(u,v) - f(u,v)
< \Delta ,$ and $c(v,u) - f(v,u) \geq \Delta$. We define the next
length function  $\tilde \ell$.
$$
\tilde \ell(u,v)  = \left\{ \begin{array}{rl}
  0 &\mbox{ if $(u,v)$ is special} \\
  \hat \ell(u,v) &\mbox{ otherwise}
       \end{array} \right.
$$
\end{definition}

Classic near-linear time blocking flow algorithm by \cite{SleatorT83}
works only for acyclic admissible graph. Note that an admissible graph
$A(G,c,f,\tilde \ell) $ may contain cycles since an edge-length can
be zero. To handle this issue, the key idea by \cite{GoldbergR98} is
to contract all strongly connected components, and run the algorithm
by \cite{SleatorT83}. To route the flow, they construct a routing flow
network inside each strongly connected components using two directed
trees for a fixed root in the componenet. One tree is for routing
in-flow, the other one is for routing out-flow from the component. The
internal routing network ensures that each edge inside is used at most
twice. Hence, by restricting at most $\Delta/4$ amount of flow, each
edge is used at most $\Delta/2$. Since each edge in the component has
length zero, it has residual capacity at least $\Delta$. So, the result of flow augmentation
respects the edge capacity.  Finally, speical edges (with the condition related to $\Delta /2$) play an
 important role to ensure that blocking flow augmentation strictly increases the distance $d_{\tilde \ell}(t)$. 

The following lemma summarizes the sketch of aforementioned algorithm.
\begin{lemma} [\cite{GoldbergR98}] \label{lem:binaryblockingflowalg}
Let $A(G,c,f,\ell)$ be an admissible graph and $m_A$ be its number
of edges. Then, there exists an algorithm that takes as input
$A$ and $\Delta > 0$, and
in $O(m_A \log (m_A))$ time, outputs a $\Delta/4\text{-or}$-blocking flow. 
We call the algorithm as $\blockingflow (A(G,c,f,\ell), \Delta)$. 
\end{lemma}

We now define the notion of {\em shortest-path flow}. Intuitively, it
is a union of shortest paths on admissible graphs. This is the flow resulting from, e.g., the Binary Blocking Flow algorithm \cite{GoldbergR98}. 

\begin{definition} [Shortest-path flow] 
Given a graph $(G,c)$ with a flow $f$, and length function $\ell$, and
let $G_f$ be the residual graph. A flow $f^*$ in $G_f$ is called
\textit{shortest-path flow}  if it can be decomposed into a set of
shortest paths under length function $\ell$, i.e., $f^* = \sum_{i=1}^{b} f_i^*$ for
some integer $b > 0$ where support$(f_i^*)$ is a shortest-path in $G_f$ under length function $\ell$. 
\end{definition}

Observe that $\blockingflow (A(G,c,f,\ell), \Delta)$ always produces
a shortest-path flow. 

From the rest of this section, we fix an augmented graph $(G',c_{G'})$
(\Cref{def:aug_graph}), and also a flow $f$. 

Given residual graph $G'_f$, and $d_{\ell}$, we can use
$\blockingflow (A(G',c_{G'},f, \tilde \ell),\Delta)$ to compute a
$\Delta/4\text{-or}$-binary blocking flow in $(G',c_{G'},f)$ in $\ot(m)$ time. 

  \cite{OrecchiaZ14} provide a slightly different binary length function such that the
algorithm in \cite{GoldbergR98} has local running time.

Our goal in next section is to output the same $\Delta/4\text{-or}$-binary
blocking flow in $G'_f$ in $\ot(\nu k)$ time using a slight adjustment
from \cite{OrecchiaZ14}.

\subsection{Local Augmented Graph and Binary Blocking Flow in Local Time} 
\label{sec:local binary blocking flow}
The goal in this section is to  compute binary blocking flow on the
residual graph of the augmented graph $(G',c_{G'})$ with a flow $f$ in
``local'' time. To ensure local running time, we cannot construct the
augmented graph $G'$ explicitly. Instead, we compute binary blocking flow from a subgraph of $G'$ based on ``absorbed'' vertices. 

\begin{definition} [Split-node-saturated set] \label{def:splitnodesauratedset}
Given a residual graph $(G',c_{G'},f)$,  let  $B_{\out}$ be the set of vertices $v \in V_{\out} \sqcup
\{ x \}$ in the residual graph $(G',c_{G'},f)$  whose edge to $t$ is
saturated. The \textit{split-node-saturated set} $B$ is defined as: $$B = B_{\out} \sqcup N_{G'}^{\out}(B_{\out})\setminus \{ t \}$$
Note that $x$ is a fixed vertex as in \Cref{def:aug_graph}.  
\end{definition}

\begin{definition} [Local binary length function] \label{def:local-length-function}
Fix a parameter $\Delta > 0$ to be selected, let $\tilde \ell$ be the
length function in \Cref{def:gr-binary-length} for the residual graph
$(G',c_{G'},f)$. For vertex $u,v$ in the residual graph,  if $u,v \in
B$, we call residual edge $(u,v)$ \textit{modern}. Otherwise, we call
residual edge $(u,v)$ \textit{classical}.  

We define \textit{local binary length} function  $\ell$:  %
$$
 \ell(u,v)  = \left\{ \begin{array}{rl}
  1 &\mbox{ if $(u,v)$ is classical} \\
  \tilde \ell(u,v) &\mbox{ otherwise}
       \end{array} \right.
$$
\end{definition}

\begin{definition} [Distance under local binary length $\ell$]
Define distance function  $d(v)$ as the shortest path distance between the source vertex $s$
and vertex $v$ in the residual graph $(G',c_{G'},f)$ under the local length function $\ell$.
\end{definition}

The following obsevations about structural properties of the residual graph
$G'_f$ follows immediately from the definition of local length function $\ell$. 

\begin{observation} \label{obs:locallength} 
For a given residual graph $(G',c_{G'},f)$, 
\begin{itemize}[noitemsep,nolistsep]
\item for any residual edge $(u,v)  \in E_{\infty,f}$ that is modern,  $\ell(u,v) = 0 $.  
\item for any residual edge $(u,v) \in E_{\deg,f} \sqcup (s,x)$,
  $(u,v)$ is classical.
\item any residual edge with length zero is modern. 
\end{itemize}
\end{observation}

\begin{definition} [Layers]
Given distance function $d$ on residual graph $(G',c_{G'},f)$, define
$L_j = \{ v \in G' \colon d(v) = j \}$ to be  the set of $j^{\text{th}}-$layer with respect to distance $d$.  Define $d_{\text{max}} = d(t)$ to be distance between $s$ and $t$ in $(G',c_{G'},f)$.
\end{definition}

The proof of the following Lemma is similar to that from
\cite{OrecchiaZ14}, but we focus on the augmented graph
$(G',c_{G'},f)$.  Recall split-node-saturated set $B$ from \Cref{def:splitnodesauratedset}.

\begin{lemma} \label{lem:layerproperties}
If $d_{max} < \infty$ and $(x,t)$ is saturated, then  we have:
\begin{enumerate}[noitemsep,nolistsep,label=(\Roman*)]
\item \label{item:layerproperties1} $d_{\text{max}} \geq 3$.
\item \label{item:layerproperties2} $L_0 = \{ s \}$.
\item \label{item:layerproperties3} $L_j \subseteq  B$ for $1 \leq j \leq d_{\text{max}}-2$.
\item \label{item:layerproperties4} $L_j \subseteq  B \cup N_{G'}^{\out}( B )$ for $ j = d_{\text{max}} -1$.
\end{enumerate}
\end{lemma}

\begin{proof}[Proof of
  \Cref{lem:layerproperties}\ref{item:layerproperties1}]
 First, $d_{\max} \geq 2$ since $(s,x_{\textin})$ and
any $(v_{\out},t) \in E_{\text{deg},f}$ is classical by
\Cref{obs:locallength}. This means $d_{\max} \geq 3$ or $d_{\max} =
2$.  We show that $d_{\max} = 2$ is not possible. Suppose for the
contradiction that $d_{\max} = 2$.  Then every intermediate edge in any $(s,t)$-path, i.e., 
$s\rightarrow x\rightarrow v_{\textin} \rightarrow v_{\out}
\rightarrow \ldots
 \rightarrow w_{\textin} \rightarrow w_{\out} \rightarrow t$, must have
zero length.  Also, the path cannot be of the form $s \rightarrow x \rightarrow
t$ since $(x,t)$ is assumed to be saturated.  In particular, $(w_{\textin}, w_{\out})$ has zero
length.  By \Cref{obs:locallength},  $(w_{\textin}, w_{\out})$ must be
modern. This edge is modern when $w_{\out} \in B$ by definition of
split-node-saturated set $B$. Therefore, $(w_{\out},t)$ is saturated,
a contradiction.  
\end{proof}

\begin{proof}[Proof of
  \Cref{lem:layerproperties}\ref{item:layerproperties2}]
The second item follows from the fact that $(s,x)$ is the only outgoing-edge from
$s$ and $(s,x)$ is classical and hence has length 1. 
\end{proof}

\begin{proof}[Proof of
  \Cref{lem:layerproperties}\ref{item:layerproperties3}]
  For $1 \leq j \leq d_{\text{max}}-2$, if $v
\in L_j$, then we consider two types of $v$. If $v$ is an out-vertex $v_{\out}$, then
$d(v_{\out}) = j \leq d_{\max}-2$. Thus, $(v_{\out},t)$ must be saturated since $d(t) = d_{\text{max}} > 
d_{\text{max}}-1 \geq d(v_{\out})+1$. Hence, $v_{\out} \in B_{out}$, which is in $B$. 

If $v$ is an in-vertex $v_{\textin}$, then there must be an out-vertex
$u_{\out}$ such that 
\begin{equation} \label{eq:in-out-vertex1} 
d(v_{\textin})= d(u_{\out}) +\ell(u_{\out}, v_{\textin})
\end{equation}
We consider two cases for $j$. We show that $v_{\textin} \in B$ for either case.
\begin{itemize} 
\item If $j = 1$, then $u_{\out}$ could also be $x$. Since $d_{\text{max}} \geq 3$,
$u_{\out}$ at $L_1$ must be saturated, meaning that $u_{\out} \in B_{\out}$. Hence, $v_{\textin}$ is an
out-neighbor of $u_{\out} \in B_{\out}$. 
\item If $j \geq 2$,  then we show that $1 \leq d(u_{\out}) \leq
  d_{\text{max}}-2$. For the upper bound $d(u_{\out}) \leq
  d_{\text{max}}-2$, rearranging \Cref{eq:in-out-vertex1}, and use the
  fact that $\ell$ is a binary function,  $\ell(u_{\out}, v_{\textin})
  \in \{ 0, 1 \}$ to get:  $$d(u_{\out}) =d(v_{\textin}) - \ell(u_{\out}, v_{\textin})  \leq d(v_{\textin}) = j \leq
d_{\text{max}}-2$$
The lower bound $d(u_{\out}) \geq 1$ follows from  $ d(v_{\textin}) = j
\geq 2$, \Cref{eq:in-out-vertex1}, and $\ell$ is binary. 

Since $1 \leq d(u_{\out}) \leq  d_{\text{max}}-2$, by the previous
discussion, $u_{\out} \in B_{\out}$. Therefore, $v_{\textin} \in B$ since
$v_{\textin}$ is the out-neighbor of $u_{\out} \in B_{\out}$. 
\end{itemize}
\end{proof}

\begin{proof}[Proof of
  \Cref{lem:layerproperties}\ref{item:layerproperties4}]

 For any $v \in L_{d_{\text{max}}-1}$, if $v
\in B$, then we are done. Now, assume that $v \not \in B$. Then,  $v$
is either an in-vertex or out-vertex. We first show that $v$ cannot be
an in-vertex. 

Suppose for contradiction that $v$ is an in-vertex  $v_{\textin} \not
\in B$, then  there must be a vertex $u_{\out}$ such that $
d(v_{\textin})= d(u_{\out}) +\ell(u_{\out}, v_{\textin})$. Since
$v_{\textin} \not \in B$, the residual edge $(u_{\out}, v_{\textin})$ is classical. Then, $\ell(u_{\out}, v_{\textin}) = 1$.   So, $$d(u_{\out}) =
  d(v_{\textin}) -  \ell(u_{\out}, v_{\textin})=  (d_{\text{max}} -1)- 1 \leq
d_{\max}-2$$ By
\Cref{lem:layerproperties}\ref{item:layerproperties3}, $u_{\out}$ is
in $B$, which means $u_{\out} \in B_{\out}$. Hence, $v_{\textin}$ is an out-neighbor of $u_{\out}
\in B_{\out}$. So, $v_{\textin} \in B$, a contradiction.   

Finally, if $v = v_{\out} \not \in  B$, then we show that $v_{\out} \in N_{G'}^{\out}( B
)$.  There exists $u_{\textin}$ such that $d(v_{\out}) = d(u_{\textin}) +
\ell(u_{\textin}, v_{\out})$.  Since $v_{\out} \not \in  B$,
$v_{\out}$ is not saturated.  Hence, $(u_{\textin},v_{\out})$ is
classical. Therefore, $u_{\textin} \in L_j$ for $j \leq d_{\text{max}}
-2$. So, $u_{\textin} \in B$ by
\Cref{lem:layerproperties}\ref{item:layerproperties3}, and
$v_{\out}$ is the out-neighbor of $u_{\textin}$. Therefore, $v_{\out} \in  N_{G'}^{\out}( B)$.  
\end{proof}

\begin{definition}[Local graph, $LG$] \label{def:aug_graph_local}
Given the augmented graph $G' = (V',E')$ and split-node-saturated set $B$, we define the \textit{local graph}
$LG(G',B) = G'[ V'' ] = (V'', E'') $ as an induced subgraph of $G'$ where 
\begin{align}\label{eq:aug_graph_local}
V'' =  B \sqcup  N_{G'}^{\out}(B) \sqcup \{ s , t \}  \quad\mbox{and}\quad E'' = E''_\nu \sqcup E''_\infty \sqcup
E''_{\text{deg}} \sqcup \{ (s,x) \}
\end{align}
where the sets in \Cref{eq:aug_graph_local} are defined as follows. 
\begin{itemize}[noitemsep,nolistsep]
\item $E''_{\nu} = \{ (v_{\textin},v_{\out}) \colon v_{\out} \in
  B_{\out} \sqcup  N_{G'}^{\out}(B) , (v_{\textin}, v_{\out}) \in E_{\nu} \}$. 
\item $E''_{\infty} = \{(v_{\out}, w_{\textin}) \colon v_{\out}  \in B_{\out}, w_{\textin} \in   V', (v_{\out}, w_{\textin}) \in E_{\infty} \}$.
\item $E''_{\text{deg}} = \{ (v_{\out},t) \colon v_{\out} \in
  B_{\out} \sqcup  N_{G'}^{\out}(B) \}$.
\end{itemize}

Using the same capacity and flow as in $G'$, the \textit{residual local graph} is $(LG(G',B), c_{LG},
f_{LG})$ where $c_{LG}$ and $f_{LG}$ are the same as $c_{G'}$
and $f_{G'}$, but restricted to the edges in $LG(G',B)$. The
local length function $\ell$ also applies to $LG(G',B)$. 
\end{definition}

\begin{lemma}  \label{lem:lg-size}Let $m'$ be the number of edges in
  $LG(G',B)$, and $n' = |V''|$ be the number of vertices in
  $LG(G',B)$. We have  $$m' \leq 4 \nu /\epsilon \quad \mbox{ and } \quad n' \leq 8 \nu/ (\epsilon k). $$
\end{lemma}
\begin{proof}

For any out-vertex $v_{\out} \in B_{\out}$, its edge to $t$ must be saturated
before it is included in $B$ with capacity of
$\text{deg}_{G}^{\out}(v)$. The edge $(s,x)$ is also an $(s,t)$-edge cut in $G'$ with capacity
$\nu/\epsilon + \nu + 1$. Hence, the maximum flow $F^*$ in $G'$ is at most
$\nu/\epsilon + \nu + 1 $.  We have $$ \sum_{v_{\out} \in B_{\out}}
\text{deg}_{G}^{\out}(v) \leq F^* \leq  \nu/\epsilon + \nu + 1.$$ 

By \Cref{lem:layerproperties} and \Cref{def:aug_graph_local}, $m' =
|E''_\nu|+ |E''_{\infty}|+|E''_{\deg}|  +1$ where $ |E''_\nu| =
|B_{\out}|+|N_{G'}^{\out}(B)| -1, |E''_{\infty}| = \sum_{v_{\out} \in B_{\out}}
\text{deg}_{G}^{\out}(v), $ and $|E''_{\deg}|  = |B_{\out}|
+|N_{G'}^{\out}(B)| $. By \Cref{def:aug_graph_local}, $ B_{\out}
\sqcup N_{G'}^{\out}(B)\subset V''$. Since $|V''| = n'$ and every
out-vertex has a corresponding in-vertex ($x$ has $s$), $|B_{\out}|
+ |N_{G'}^{\out}(B)| \leq n'/2 \leq \sum_{v_{\out} \in B_{\out}} \text{deg}_{G}^{\out}(v)$. So,
$$ m'  \leq 2 \sum_{v_{\out} \in B_{\out}} \text{deg}_{G}^{\out}(v) 
\leq 2(\nu/\epsilon + \nu + 1)  \leq  4 \nu /\epsilon. $$ 

To compute $n'$, note that each $v_{\out}$ has at least
$d_{\text{min}}^{\out} \geq k $ edges. Therefore, the number of vertices including $v_{\textin}$ is at most
$n' \leq 2(m'/d_{\text{min}}^{\out})  \leq 2 m'/k  \leq  8
\nu/(\epsilon k) $.  
\end{proof} 

\begin{corollary} \label{cor:construct-lg-local}
Given a residual graph $ (G',c_{G'},f) $ and split-node-saturated
set  $B$, and a pointer to vertex $x$, we can
construct  $(LG(G',B), c_{LG}, f)$  in $O(m') = O(\nu / \epsilon)$ time. 
\end{corollary}

The proof of the following Lemma is a straightforward modification from \cite{OrecchiaZ14}. 
\begin{lemma}\label{lem:localbinaryblockingflow}
Given the local length function $\ell$ on both residual augmented
graph $(G', c_{G'}, f) $
and residual local graph $(LG, c_{LG}, f_{LG}) = (V'', E''_f,
c_{LG,f})$ (Recall $f_{LG}$ from \Cref{def:aug_graph_local}). Let $f_1$ be the output of $\blockingflow
(A(G',c_{G'}, f, \ell), \Delta)$. Let $f_2$ be the output of\\
$\blockingflow (A(LG,c_{LG}, f_{LG}, \ell), \Delta)$. Then, 

\begin{itemize}[noitemsep,nolistsep] 
\item  $ f_1 = z(f_2)$ where  
 $$   
   z (f_2) (e) = \left\{ \begin{array}{rl}
   0 &\mbox{ if $e \not \in E''_f$.} \\ 
   f_2 (e) &\mbox{ otherwise}
  \end{array} \right. 
$$ i.e.,  $f_1$ and $f_2$ coincide. 
 
\item $\blockingflow(A(LG,c_{LG}, f_{LG}, \ell), \Delta)$  takes $\ot(\nu /\epsilon)$  time. 
\end{itemize}  
\end{lemma}
\begin{proof} 
 We focus on proving the first item. For notational convenience,
 denote $G'_f = (G', c_{G'}, f)$, and $LG_f = (LG, c_{LG}, f_{LG})$.  We show that there is no
 $(s,t)$-path in $G_f'$ if and only if there is no $(s,t)$-path in
 $LG_f$.   The forward direction follows immediately from the fact that
 $LG_f$ is a subgraph of $G'_f$.  Next, we show the backward
 direction. Let $U$ be a subset of vertices in graph $LG_f$
 such that $s \in U$ and $t \not \in U$ and there is no edge between
 $U$ and $V_{LG_f} \setminus U$. By \Cref{def:aug_graph_local}, $U \subseteq
 V_{LG_f}  = B \sqcup N_{G'}^{\out}(B) \sqcup \{ s \}$.  In fact, $U
 \subseteq B \{ s \}$  since   vertices in $N_{G'}^{\out}(B)$ have residual edges to
 sink $t$ with positive residual capacity by the construction of
 $LG_f$. Now, we claim that all edges at the boundary of $B \sqcup \{ s \}$ in $G'$ and
 $LG$ are the same. Indeed, all edges at the
 boundary of $B \sqcup \{ s \}$  have the form $(u,v)$ where $u \in B$
 and $v \in N_{G'}^{\out}(B)$ and $B \sqcup N_{G'}^{\out}(B) \subseteq
 V''$ in $LG_f$. Furthermore, there is no $(U,t)$ path in $LG_f$ where $U \ni s$. Therefore, there is no $(s,t)$ path in $G'_f$.  
 
For the rest of the proof, we assume that there is an $(s,t)$ path,
i.e., $d(t) = d_{\text{max}} < \infty$. 

We claim that a flow $f^*$ in $G'_f$ is shortest-path flow if
and only if the same flow $f^*$ when restricting to edges in $LG_f$ is
shortest-path flow. 

We show the forward direction. If $f^*$ in $G'_f$ is
shortest-path flow, then by definition of shortest-path
flow, the support of $f^*$ contains only vertices with $d(v) <
d_{\text{max}}$ and  $t$. By \Cref{lem:layerproperties}, 
$$ \{ s \} \sqcup L_1 \sqcup \ldots \sqcup L_{d_{\text{max}}-1} \sqcup
\{ t \} \subseteq \{s,t\} \sqcup  B \sqcup N_{G'}^{\out}(B)$$

We show that support of $f^*$ form a subgraph of $LG$. The edges are
either between consecutive layers $L_i, L_{i+1}$ or within a layer. We
can limit the edges using \Cref{lem:layerproperties}.  From $s$ to vertices in $L_1$, there is only one edge
$(s,x)$. Edges from $L_i$ to $L_{i+1}$ for $1 \leq i \leq
d_{\text{max}}-2$ must be of the form $\{(v_{\textin},v_{\out}),
(v_{\out},v_{\textin}),  \colon v_{\textin}, v_{\out}  \in B \}$
or  $\{ (v_{\textin}, w_{\out}) \colon w_{\out} \in
N_{G'}^{\out}(B), w_{\textin} \in   B\}$. From
$L_{d_{\text{max}}-1}$ to $L_{d_{\text{max}}}$, the edge must be of the form $\{ (v_{\out},t) \colon v_{\out} \in B_{\out} \}$. If the
edges are within a layer, then they must be modern since their length
is zero. This has the form of $\{ (u,v) \in E' \colon u, v \in B
\}$.  

Since the support of $f^*$ form a subgraph of $LG$, we can restrict
$f^*$ to the graph $LG$, and we are done with the forward direction of
the claim. 

We show the backward direction of the claim. Let $f'$ be a
shortest-path flow in $LG_f$. We can extend $f'$ to be the flow
in $G'_f$ by the function $z(f')$.  $$   
   z (f') (e) = \left\{ \begin{array}{rl}
   0 &\mbox{ if $e \not \in E''_f$.} \\ 
   f' (e) &\mbox{ otherwise}
  \end{array} \right. 
$$ 
 The support of the flow function $z(f')$ in $G'_f$ contains
vertices in  contains only vertices with $d(v) < d_{\text{max}}$ and
$t$ since $f'$ is the shortest-path flow.  Therefore, $z(f')$ is
the shortest-path flow in $G'_f$, and we are done with the
backward direction of the claim. 

Finally, the first item of the lemma follows since $\blockingflow$
outputs a shortest-path flow. 
 
The running time for the second item follows from the fact that number 
of edges $m'$ in $LG(G',B) $ is $O(\nu/\epsilon )$ by \Cref{lem:lg-size}.  By
\Cref{lem:binaryblockingflowalg}, $\blockingflow (A(LG,c_{LG}, f, \ell), \Delta)$  can be
computed in  $\ot(m') = \ot(\nu/\epsilon)$  time. 
\end{proof}

\subsection{Local Goldberg-Rao's Algorithm for Augmented Graph}

\begin{theorem} \label{lem:local-gold-berg-rao}
Given graph $G$, we can compute the $(s,t)$ max-flow in $G'$  in $\ot
(\nu^{3/2}/(\epsilon^{3/2} \sqrt{k}))$ time. 
\end{theorem}

\begin{algorithm}[H]
\KwIn{$x \in V, \nu, k$} 
\KwOut{maximum $(s,t)$-flow  and its corresponding minimum
  $(s,t)$-edge-cut in $G'$}  
\BlankLine
Let $G'$ be an \textit{implicit} augmented graph on $G$. \tcp*{No
  need to construct explicitly.} 
$\Lambda \gets \sqrt{8\nu/(\epsilon k)} $ \;
$F \leftarrow 2\nu k+ \nu + 1 -\text{deg}_{G}^{\out}(x) $ \tcp*{$F$ is an upper bound on
  $(s,t)$-flow value in $G'$.} 
\lIf{ $F  \leq 0 $} { the minimum $(s,t)$-edge-cut is $(s,x)$, and return. } 
$f \leftarrow $ a flow of value $\text{deg}_{G}^{\out}(x)$ through
$s-x-t$ path.  \;
$B \leftarrow \{ x \} \sqcup N_{G'}^{\text{out}}(x)$ \tcp*{a set of
  saturated vertices and out-neighbors.}  
\While{ $F \geq 1$} 
{
  $\Delta \leftarrow F/(2\Lambda)$ \;
  \For{ $i\leftarrow 1$ \KwTo $5\Lambda$} {
    $LG \leftarrow $ local subgraph of $G'$ given $B$. \tcp*{ see  \Cref{def:aug_graph_local}, \Cref{cor:construct-lg-local}}
    $\ell \gets $ local length function on current flow $f$.\;
    $f \leftarrow f + \blockingflow(A(LG,c_{LG}, f, \ell), \Delta)$. \;
    $C \gets $ vertices in $N_{G'}^{\text{out}}(B)$ whose edges to sink
    are saturated in the new flow. \;
    $B \gets B \sqcup C \sqcup N_{G'}^{\text{out}}(C)$ \;
  }
  $F \gets F/2$ \; 
} 
\Return{maximum $(s,t)$-flow $f$ and its corresponding minimum
  $(s,t)$-edge-cut $A$ in $G'$.}
\caption{LocalFlow$(G, x,\nu,k)$}
\label{alg:localgoldbergrao}
\end{algorithm}

\textbf{Correctness.} We show that $F$ is the upper bound on the
maximum flow value in $G'_f$. We use induction on inner loop. Before
entering the inner loop for the first time, $F$ is set to be the value
of $(s,t)$ edge minus $\text{deg}_{G}^{\out}(x)$. Since $F$ is
positive, then $G_f$ has valid maximum flow upper bound $F$.  Now, we
consider the inner loop. After $5\Lambda$ times, either 
\begin{itemize}[noitemsep, nolistsep]
\item we find a flow of value $\Delta/4$ at least $4\Lambda$ times, or 
\item we find a blocking flow at least $\Lambda$ times. 
\end{itemize}
If the first case holds, then we increase the flow by at least
$\geq (\Delta/4)(4\Lambda) = F/2$. Hence, the flow $F/2$ is the valid
upper bound. For the second case, we need the following Lemma whose
proof is essentially the same as the original proof of Goldberg-Rao's algorithm:
\begin{lemma} \label{lem:block-flow-increase-dist}
A flow augmentation does not decrease the
distance $d(t)$. On the other hand, a blocking flow augmentation strictly increases $d(t)$. 
\end{lemma}
\begin{proof}[Proof Sketch]
The only issue for a blocking flow augmentation is that $s-t$ distance in
residual graph may not increase if an edge length decrease from 1 to 0. This happens when
such an edge is modern since classical edges have a constant length of
1. The proof that modern edges do not have the issue follows exactly
from the classic Gaoberg-Rao algorithm \cite{GoldbergR98} using the
  notion of special edges.
\end{proof} 

If the second case holds, we claim: 
\begin{claim}
If we find a blocking flow at least $\Lambda$ times, then there exists an $(s,t)$-edge cut of
capacity at most $\Delta \Lambda = F/2$, which is an upper bound of the remaining
flow to be augmented. 
\end{claim}
\begin{proof}
 
Before entering the inner loop for the first time, by
\Cref{lem:layerproperties}, $d(t) = d_{\text{max}} \geq 3$. After
$\Lambda$ blocking flow augmentation, $d(t) \geq 3 + \Lambda$ by
\Cref{lem:block-flow-increase-dist}. Since the $\Delta$-blocking flow in
$G'$ on $B$ and $LG$ coincide by \Cref{lem:localbinaryblockingflow}, we
always get the correct blocking flow augmentation. 

Let $L_0, L_1, \ldots, L_{d_{\text{max}}}$ be the layers of vertices
with distance $0, 1, \ldots, d_{\text{max}} = d(t) \geq  3 +
\Lambda $. We focus on edges between layers $L_i, L_{i+1}$ for $1 \leq
i \leq d_{\text{max}}-2$. By \Cref{lem:layerproperties},  any two vertices $v_1 \in L_i,
v_2 \in L_{i+1}$ must be in $B$. Therefore, by definition of local
length function $\ell$, all edges between $L_i, L_{i+1}$ must be
modern. Since any edge between $L_i, L_{i+1}$ is modern, and has
length 1, it must be of the form $(v_{\textin}, v_{\out})$ or
$(v_{\out}, v_{\textin})$ with residual capacity $\leq \Delta$ by
definition of local length function (\Cref{def:local-length-function}).

Since there are at least $\Lambda  =  \sqrt{8\nu/(\epsilon k)} \geq \sqrt{n'}$
layers $L_i$ (By \Cref{lem:lg-size}) where $1 \leq
i \leq d_{\text{max}}-2$, by counting argument, there must be a layer
$L'$ such that $|L'| \leq \sqrt{n'}$. 

Next, for any vertex $v \in L'$, $v$ has either a single outgoing edge or a single incoming edge by
construction of $G'$ since this edge must be of the form  $(v_{\textin}, v_{\out})$ or
$(v_{\out}, v_{\textin})$. 

Therefore, we find an $(s,t)$-edge-cut consisting of the single
incoming-or-outgoing edge from each node in $L'$. This cut has
capacity at most $\Delta \sqrt{n'} \leq \Delta   \sqrt{8\nu/(\epsilon
  k)}   = \Delta \Lambda = F/2$. 
\end{proof}

The correctness follows since at the end of the loop we have $F < 1$. 

\textbf{Running Time.}  By \Cref{lem:localbinaryblockingflow}, we can
compute $\Delta$-blocking flow in $LG$ with local binary length function
$\ell$ in $\ot(\nu/\epsilon)$ time. The time already includes the time to
read $LG$. The number of such computations is
$O(\Lambda \log (\nu /\epsilon) ) = O(\sqrt{\nu/(\epsilon k)} \log (m) ) = \ot(
\sqrt{\nu/(\epsilon k)})$.  So the total running time is
$\ot(\nu^{3/2} / (\epsilon ^{3/2} k^{1/2}))$. This
completes the proof of \Cref{lem:local-gold-berg-rao}.

\subsection{Proof of \Cref{thm:local-vertex-connectivity}}

\begin{proof}[Proof of \Cref{thm:local-vertex-connectivity}]
Given $G, x, \nu, k, \epsilon$,  by \Cref{lem:local-gold-berg-rao}, we compute the minimum
$(s,t)$-edge-cut $C^*$ in $G'$ in $ \ot(\nu^{3/2} / (\epsilon ^{3/2} k^{1/2})$ time. If the
edge-cut $C^*$ has capacity $> \nu/\epsilon + \nu$, then by \Cref{lem:aug_graph_properties}\ref{item:aug_graph_properties_one}, we can
output $\perp$. Otherwise, $C^*$ has capacity at most $\nu/\epsilon +
\nu$, by \Cref{lem:aug_graph_properties}\ref{item:aug_graph_properties_two}, we can
output the separation triple $(L,S,R)$ with the properties in
\Cref{lem:aug_graph_properties}\ref{item:aug_graph_properties_two}. 
\end{proof}                    
	\section{Vertex Connectivity via Local Vertex Connectivity}

\begin{theorem} [Exact vertex connectivity]
\label{thm:exact_vertex_connectivity}
There exist randomized (Monte Carlo) algorithms that take as inputs a
graph $G$, integer $0 < k < O(\sqrt{n})$, and in  $ \ot ( m + k^{7/3}n^{4/3} ) $ time for undirected graph (and
in $\ot (\min( km^{2/3}n, km^{4/3}))  $ time for directed graph) can
decide w.h.p. if $\kappa_G \ge k$.   If $\kappa_{G}<k$,	then the
algorithms also return the corresponding vertex-cut. 

\end{theorem}


We define the function $T(k,m,n)$ as 
\begin{equation} \label{eq:exact-directed-time}
T(k,m,n)= \left\{ \begin{array}{rl}
  \min(m^{4/3}, nm^{2/3}k^{1/2},\\ mn^{2/3+o(1)}/k^{1/3}, \\ n^{7/3+o(1)}/k^{1/6} )    &\mbox{ if $ k \leq n^{4/5}$,} \\
   n^{3+o(1)}/k    &\mbox{ if $k > n^{4/5}.$ }
  \end{array} \right. 
\end{equation}

\begin{theorem} [Approximate vertex connectivity] \label{thm:approx_vertex_connectivity}
There exist randomized (Monte Carlo) algorithms that take as inputs a
graph $G$, an positive integer $k$, and positive real $\epsilon < 1$,
and in $\ot ( m+\poly(1/\epsilon) \min( k^{4/3} n^{4/3}, k^{2/3}
n^{5/3+o(1)}, n^{3+o(1)}/k)) $ 
time for undirected graph 
(and in $\ot ( \poly(1/\epsilon) T(k,m,n))$ time for directed graph
where $T(k,m,n)$ is defined as in \Cref{eq:exact-directed-time})
w.h.p. return  a vertex-cut with size at most
$(1+O(\epsilon))\kappa_G$ or cerify  that $\kappa_G \geq k$. 

\end{theorem}

This section is devoted to proving \Cref{thm:exact_vertex_connectivity}. and \Cref{thm:approx_vertex_connectivity}.
\subsection{Vertex Connectivity Algorithms} 
\begin{algorithm}[H]
\KwIn{Sampling method, LocalVC, G = (V,E), k, a, $\epsilon$} 
\KwOut{a vertex-cut $U$ such that $|U| \leq k$ or a symbol $\perp$.}     
\BlankLine
If undirected, replace $E = \{ (u,v) , (v,u) \colon (u,v) \in
E({H_{k+1}}) \}$ where $H_{k+1}$ as in \Cref{thm:sparsification}. \;

\If{\normalfont{Sampling} method $ = $ vertex} 
{
  \For{ $i \leftarrow 1$ \KwTo $n/(\epsilon a)$   \normalfont{(use} $n/a$ for exact version) }
{ 
  Sample a random pair of vertices $x, y \in V$. \;
  \lIf{$k$ \normalfont{is not specified}} 
  {
	compute approximate $\kappa_G(x,y)$. 
  }
  \If{$\kappa_G(x,y) \leq (1+\epsilon)k$}  
  { 
    \Return{\normalfont{the} corresponding $(x,y)$-vertex-cut $U$.}
  }
}

}

\If{\normalfont{Sampling} method $ = $ edge} 
{ 
\For{ $i \leftarrow 1$ \KwTo $m/(\epsilon a)$   \normalfont{(use} $m/a$ for exact version)     } 
{
  Sample a random pair of edges $(x_1,y_1), (x_2,y_2) \in E$. \;
  \If{$k$ \normalfont{is not specified}} 
  {
	compute approximate $\kappa_G(x_1,y_2), \kappa_G(x_1,x_2), \kappa_G(y_1,x_2),
    \kappa_G(y_1,y_2)$. 
  }
  
  \If{$\min(\kappa_G(x_1,y_2), \kappa_G(x_1,x_2), \kappa_G(y_1,x_2),
    \kappa_G(y_1,y_2)) \leq (1+\epsilon)k$}  
  { 
    \Return{\normalfont{the} corresponding $(x,y)$-vertex-cut $U$.}
  }

}
}

\If{\normalfont{LocalVC is not specified}} {
Let $x^*,y^*$ be vertices with minimum $\kappa_G(x^*,y^*)$ computed so far.\;
Let $W$ be the vertex-cut corresponding to $\kappa_G(x^*,y^*)$ \;
Let $v_{\text{min}}, u_{\text{min}} $ be the vertex with the minimum out-degree in $G$ and $G^R$ respectively. \;
\Return{ \normalfont{The smallest set among }$ \{ W,
N_G^{\out}(v_{\text{min}}),N_{G^R}^{\out}(u_{\text{min}}) \} $.} }

Let $\cL = \{  2^{\ell} \colon 1 \leq \ell \leq \lceil  \log_2 a \rceil $, and $\ell \in \mathbb{Z} \}$. \;

\If{\normalfont{Sampling method} $ = $ vertex} 
{
\For{$ s \in \cL$}
{
   \For{ $i \gets 1 $ \KwTo $n/s$}  
   {
      Sample a random vertex $x \in V$. \;
      Let $\nu \leftarrow O(s(s+k))$.  \;
      \If{\normalfont{LocalVC}$(G,x,\nu,k,\epsilon)$ or \normalfont{LocalVC}$(G^R,x,\nu,k,\epsilon)$  outputs a vertex-cut $U$}  {\Return{$U$.} } 
    } 
}

}

\If{\normalfont{Sampling method} $ = $ edge} 
{
\For{$ s \in \cL$}
{
   \For{ $i \gets 1 $ \KwTo $m/s$}  
   {
      Sample a random edge $(x,y) \in E$. \;
      Let $\nu \gets O(s),$ and $\mathcal{G} = \{ G, G^R \}$. \;
	  \For{$H \in \mathcal{G}, z \in \{x,y\}$} 
	  {
		\If{\normalfont{LocalVC}$(H,z,\nu,k,\epsilon)$  outputs a vertex-cut $U$.}
		{
			\Return{$U$.}
		}
	  }
    } 
}
}

\Return{$\perp$.}

\caption{VC$(\text{Sampling method}, \text{LocalVC}, \kappa(x,y); G, k, a, \epsilon)$}
\label{alg:vcframework}
\end{algorithm}

 \subsection{Correctness}

We can compute approximate vertex connectivity by standard binary search on $k$ with the decision problem. We focus on correctness of \Cref{alg:vcframework} for approximate version. For exact version, the same proof goes through when we use $\epsilon = 1/(2k)$, and $\kappa_G \leq \sqrt{n}/2$. Let $\Delta = \min(n/(1+\epsilon), (m/(1+\epsilon))^{1/2}))$.  For the purpose of analysis of the decision problem, we assume the followings.
\begin{assumption}  \label{ass:correctness}
If $k$ is specified in \Cref{alg:vcframework}, then
\begin{enumerate}[noitemsep,nolistsep,label=(\Roman*)]
	\item \label{item:correct_one} $\dmin \geq k$.
	\item \label{item:correct_two} $k \leq \Delta$. We use $k \leq \sqrt{n}/2$ for exact vertex connectivity. 
	\item \label{item:correct_three} Local conditions  in \Cref{thm:local-vertex-connectivity} are satisfied. We use exact version of local conditions for exact vertex connectivity. 
\end{enumerate}
\end{assumption}

We justify above  assumptions. For
\Cref{ass:correctness}\ref{item:correct_one}, if it does not hold,
then we can trivially output the neighbors of the vertex with minimum
degree and we are done. 

For \Cref{ass:correctness}\ref{item:correct_two}, if we can verify that $\kappa_G \geq k = \Delta$, then in \Cref{sec:high_vertex_conn} we show that the out-neighbors of the vertex with minimum out-degree is an approximate solution. For exact vertex connectivity, we either find a minimum vertex-cut or verify that $\kappa_G \geq \sqrt{n}/2$. For \Cref{ass:correctness}\ref{item:correct_three}, we can easily verify the parameters $a$ and $\nu,k,\epsilon$ supplied to the LocalVC.

Ignoring running time, we classify \Cref{alg:vcframework} into four
algorithms depending on sampling edges or vertices, and using LocalVC
or not. We omit edge-sampling without LocalVC since the running time
is subsumed by vertex-sampling counterpart.  We now prove the correctness for each of them.  

\subsubsection{High Vertex Connectivity} \label{sec:high_vertex_conn}
 We show that if we can verify that the graph has high vertex connectivity, then we can simply output the out-neighbors of the vertex with minimum out-degree to obtain an $(1+\epsilon)$-approximate solution. 

\begin{proposition}
If $\kappa_G \geq \Delta$, then $|\dmin| \leq (1+\epsilon)\kappa_G$.
\end{proposition}
\begin{proof}
We first show that if $\kappa_G \geq (m/(1+\epsilon))^{1/2}$, then $\kappa_G \geq n/(1+\epsilon)$. Since $\kappa_G \geq (m/(1+\epsilon))^{1/2}$, we have $\kappa_G^2 \geq m/(1+\epsilon)$. Therefore, we obtain 
$ \kappa_G \dmin \geq \kappa_G^2 \geq m/(1+\epsilon) \geq n \dmin /(1+\epsilon).$
The first inequality follows from \Cref{obs:kappa-degree}, which is $\dmin \geq \kappa_G$. The second inequality follows from above discussion. The third inequality follows from each vertex has at least $\dmin$ edges. Therefore, $ \kappa_G \geq n/(1+\epsilon)$. 

Now, we show that if $\kappa_G \geq n/(1+\epsilon)$, then $ \dmin \leq (1+\epsilon)\kappa_G$. We have 
$ (1+\epsilon)\kappa_G \geq n \geq \dmin \geq \kappa_G.$ The first inequality follows from the condition above. The second inequality follows from size of vertex cut is at most $n$. The third inequality follows from \Cref{obs:kappa-degree}.
\end{proof}
 
\subsubsection{Edge-Sampling with LocalVC}
\begin{lemma} \label{lem:edge_sampling_localvc_correct}
 \Cref{alg:vcframework} with edge-sampling, and  $\operatorname{LocalVC }$ outputs correctly w.h.p. a vertex-cut of
size $\leq (1+\epsilon)k$ if $\kappa_G \leq k$, and a symbol $\perp$ if $\kappa_G > k $.    
\end{lemma}

We describe notations regarding edge-sets from a separation triple
$(L,S,R)$ in $G$. Let $ E^*(L,S) =  E(L,L) \sqcup E(L,S) \sqcup E(S,L)$,
and $ E^*(S,R) = E(R,R) \sqcup E(S,R) \sqcup E(R,S)$.   

\begin{definition}  [$L$-volume, and $R$-volume of the separation
  triple] \label{def:lrvol}
For a separation triple $(L,S,R)$, we denote $\vol_G^*(L) = \sum_{v \in L} \deg_G^{\out}(v) + |E(S,L)|$ and
$\vol_G^*(R) = \sum_{v \in R} \deg_G^{\out}(v) + |E(S,R)|$.
\end{definition}

It is easy to see that $\vol_G^*(L)  = |E^*(L,S)|$ and $\vol_G^*(R) = |E^*(S,R)|$. 

The following observations follow immediately from the definition of
$E^*(L,S)$ and $E^*(S,R)$, and a separation triple $(L,S,R)$.

\begin{observation} \label{obs:edge_star}
We can partition edges in $G$  according to $(L,S,R)$ separation triple as $$E
=  E^*(L,S) \sqcup E(S,S) \sqcup  E^*(S,R)$$ 
And,
\begin{itemize}[nolistsep,noitemsep]
\item For any edge $ (x,y) \in E^*(L,S), x \in L $ or $ y \in L$.
\item For any edge $ (x,y) \in E^*(S,R), x \in R $ or $ y \in R$.
\end{itemize}
Furthermore, $$ m = \vol_G^*(L)  + |E(S,S)| + \vol_G^*(R)$$
\end{observation}

We proceed the proof. There are three cases for the set of all
separation triples in $G$.  The first case is there exists a separation triple $(L,S,R)$ such
that $|S| \leq k, \vol^*_G(L) \geq a, \vol^*_G(R) \geq a$  We show that w.h.p.  \Cref{alg:vcframework} outputs a vertex-cut of size at most $(1+\epsilon)k$. 

\begin{lemma} \label{lem:balanced_edge}
If $G$ has a separation triple $(L,S,R)$ such that $ |S| \leq k,
\vol^*_G(L) \geq a, \vol^*_G(R) \geq a $, then w.h.p. \Cref{alg:vcframework} outputs a vertex-cut of size at most $(1+\epsilon)k$. 
\end{lemma}
\begin{proof}

We show that the first  loop (with edge-sampling method) of \Cref{alg:vcframework} finds a vertex-cut of size at most $(1+\epsilon)k$. 

We sample two edges randomly $e_1 = (x_1,y_1), e_2 = (x_2,y_2) \in
E$. The probability that $e_1 \in  E^*(L,S)$ and $e_2 \in 
E^*(S,R)$ is  $$ P(e_1 \in  E^*(L,S), e_2 \in  E^*(S,R)) = P(e_1
\in  E^*(L,S)) P( e_2 \in  E^*(S, R)) $$ This follows from the
two events are independent. 

By \Cref{ass:correctness}\ref{item:correct_two},  $k \leq \Delta$, which means $k^2 \leq m/(1+\epsilon)$. For exact vertex connectivity, we have $k^2 \leq n/4 \leq m/4$.  For generality, we denote $k^2 \leq m/c$. We use $c = 1+\epsilon$  for the approximate vertex connectivity, and $c = 4$ for exact version.

We claim $\vol^*_G(L) + \vol^*_G(R) = \Omega((1-1/c)m)$. Indeed, by \Cref{obs:edge_star}, $\vol^*_G(L) + \vol^*_G(R) = m - |E(S,S)|$, and we have $|E(S,S)| \leq k^2 \leq m/c$.

If $\vol^*_G(R) = \Omega((1-1/c) m)$, then  $ P(e_1 \in  E^*(L,S), e_2 \in  E^*(S,R)) = P(e_1 \in  E^*(L,S)) P( e_2 \in  E^*(S, R)) \geq \vol^*_G(R) a/ m^2 = \Omega( (1-1/c)a/m) $. 
Otherwise, $\vol^*_G(L) = \Omega((1-1/c)m)$. Similarly, we get  $ P(e_1 \in  E^*(L,S), e_2 \in  E^*(S,R)) = \Omega( (1-1/c)a/m) $.

Therefore, it is enough to sample $O(m/(\epsilon a))$ times ($O(m/a)$ times for the exact vertex connectivity) to get w.h.p. at least
one trial where $e_1 = (x_1, y_1) \in  E^*(L,S), e_2 = (x_2,y_2)
\in  E^*(S,R)$. From now we assume, $e_1 = (x_1, y_1) \in  E^*(L,S), e_2 = (x_2,y_2) \in  E^*(S,R)$.  

Finally, we  show that the first loop of \Cref{alg:vcframework}
outputs a vertex-cut of size at most $k$. By \Cref{obs:edge_star}, 
at least one vertex in $(x_1,y_1)$   is in $L$, and at least on vertex in  $(x_2,y_2)$  is in
$R$. Therefore, we find a separation triple corresponding to $\min(\kappa_G(x_1,y_2), \kappa_G(x_1,x_2), \kappa_G(y_1,x_2), \kappa_G(y_1,y_2)) \leq (1+\epsilon)k$. 
\end{proof}
 
The second case is there exists a separation triple $(L,S,R)$ such that
$|S| \leq k$ and $\vol^*_G(L) < a   $ or $\vol^*_G(R) < a$. We show that
w.h.p. \Cref{alg:vcframework} outputs a vertex-cut of size at most
$(1+\epsilon)k$. 

\begin{lemma}\label{lem:imbalanced_edge}
If $G$ has a separation triple $(L,S,R)$ such that  $|S| \leq k$ and $\vol^*_G(L) < a   $ or $\vol^*_G(R) < a$,
 then w.h.p. \Cref{alg:vcframework} outputs a vertex-cut of size at most $(1+\epsilon)k$. 
\end{lemma}
\begin{proof}

We show that the second loop (LocalVC with edge-sampling mode) of \Cref{alg:vcframework} finds a vertex-cut of size at most $(1+\epsilon)k$. 

We focus on the case $\vol^*_G(L) < a$. The case $\vol^*_G(R)  < a$ is similar, except that we need to compute local vertex connectivity on the reverse graph instead. 

We show that w.h.p. there is an event $e = (x,y) \in  E^*(L,S)$. Since $\vol^*_G(L) < a$,  there exists an integer $\ell$  in range  $1 \leq \ell \leq \lceil \log_2 a
\rceil$ such that $ 2^{\ell -1} \leq \vol^*_G(L) \leq s^{\ell}$. That is, $s/2
\leq \vol^*_G(L) \leq s$ for $s = 2^{\ell}$.  The probability that $e
\in E^*(L,S)$ is $\vol^*_G(L)/m \geq s/(2m)$. Hence, it is enough to
sample $O(m/s)$ edges to get an event $e \in E^*(L,S)$ w.h.p. 

From now we assume that $\vol^*_G(L) \leq s$ and that $e = (x,y) \in
E^*(L,S)$. By \Cref{def:lrvol}, $\vol_G^*(L) =\sum_{v \in L}
\deg_G^{\out}(v) + |E(S,L)| \leq s$. Therefore, $\vol_G^{\out}(L) = \sum_{v \in L} \deg_G^{\out}(v) \leq s$. 

By \Cref{obs:edge_star}, $x \in L$ or $y \in L$. We assume
WLOG that $x \in L$ (\Cref{alg:vcframework} runs LocalVC on both $x$
and $y$).  

Hence, we have verified the following conditions for the parameters $x, \nu,
k$ for LocalVC$(G,x,\nu,k)$:
\begin{itemize}[nolistsep, noitemsep]
\item  Local conditions are satisfied by \Cref{ass:correctness}\ref{item:correct_three}. %
\item $x \in L$. 
\item $|S| \leq k$.
\item $\vol_G^{\out}(L) \leq \nu$ and we use $\nu =  s$. 
\end{itemize}
  By \Cref{thm:local-vertex-connectivity}, LocalVC outputs a vertex-cut of size at most $(1+\epsilon)k$.
\end{proof}

The final case is when every separation triple $(L,S,R)$ in $G$, $|S| >
k$. In other words, $\kappa_G > k$. If \Cref{alg:vcframework} outputs a vertex-cut, then it is a $(1+\epsilon)$-approximate vertex-cut. Otherwise,  \Cref{alg:vcframework} outputs $\perp$ correctly.

%

The proof of \Cref{lem:edge_sampling_localvc_correct} is complete since
\Cref{lem:edge_sampling_localvc_correct} follows from
\Cref{lem:balanced_edge}, \Cref{lem:imbalanced_edge}, and the case $\kappa_G > k$ corresponding the three cases of the set of separation triples in $G$.

\subsubsection{Vertex-Sampling with LocalVC}
\begin{lemma} \label{lem:vertex_sampling_localvc_correct}
 \Cref{alg:vcframework} with vertex-sampling, and $\operatorname{LocalVC}$ outputs correctly w.h.p. a vertex-cut of
size $\leq (1+\epsilon)k$ if $\kappa_G \leq k$, and a symbol $\perp$ if $\kappa_G > k $.    
\end{lemma}

We consider three cases for the set of all separation triples in $G$.
The first case is there exists a separation triple $(L,S,R)$ such
that $|S| \leq k, |L| \geq a,$ and $ |R| \geq a$.  We show that w.h.p. 
\Cref{alg:vcframework} outputs a vertex-cut of size at most $(1+\epsilon)k$. 

\begin{lemma} \label{lem:balanced_v}
If $G$ has a separation triple $(L,S,R)$ such that $|S| \leq k, |L|
\geq a,$ and $ |R| \geq a$ Then w.h.p. \Cref{alg:vcframework} outputs a
vertex-cut of size at most $(1+\epsilon)k$. 
\end{lemma}
\begin{proof}

We show that the first loop of \Cref{alg:vcframework} finds a
vertx-cut of size at most $(1+\epsilon)k$. 

We sample two vertices independently $x,y \in V$. Since 
 two events $x \in L$ and $y \in R$ are independent, the probability that
$x \in L$ and $y \in R$ is $P(x \in L, y \in R) = P(x \in L) P(y \in R)$. 

By \Cref{ass:correctness}\ref{item:correct_two},  $k \leq \Delta$, which means $k \leq n/(1+\epsilon)$. For exact vertex connectivity, we have $k \leq \sqrt{n}/2 \leq n/2$.  For generality, we denote $k \leq n/c$. We use $c = 1+\epsilon$  for the approximate vertex connectivity, and $c = 2$ for exact version. 

Since $k \leq n/c$, we have $|L|+|R| = n - |S| \geq n - k \geq n - n/c = (1-1/c)n$. 
If $|R| = \Omega((1-1/c)n)$, then $P(x \in L, y \in R) = P(x \in L) P(y \in
R) \geq |R| a/n^2 = \Omega( (1-1/c)a/n) $. Otherwise, $|L| = \Omega((1-1/c)n)$, and with similar argument we get $P(x \in L, y \in R) = \Omega( (1-1/c)a/n) $.

Therefore, it is enough to sample $O(n/(a\epsilon))$ times (and $O(n/a)$ times for exact version) to get at least one trial corresponding to the event
$x \in L$ and $y \in R$ w.h.p. With that event, we can find a
separation triple corresponding to $\kappa(x,y) \leq (1+\epsilon)k$.
\end{proof}

The second case is there exists a separation triple $(L,S,R)$ such
that $|S| \leq k$ and $|L| < a$ or $|R| < a$. We show that w.h.p. 
\Cref{alg:vcframework} outputs a vertex-cut of size at most $(1+\epsilon)k$.

\begin{lemma}\label{lem:imbalanced_v}
If $G$ has a separation triple $(L,S,R)$ such that 
$|S| \leq k$ and $|L| < a$ or $|R| < a,$ then w.h.p. \Cref{alg:vcframework} outputs a vertex-cut of size at most $(1+\epsilon)k$. 
\end{lemma}
\begin{proof}

We show that the second loop (LocalVC with vertex sampling method) of \Cref{alg:vcframework} finds a vertex-cut of size at most $(1+\epsilon)k$. 

We focus on the case $|L| < a$. The case $|R| < a$ is similar,
except that we need to compute local vertex connectivity on the reverse graph instead. 

We show that w.h.p, there is an event $x \in L$. Since $|L| < a$,  there exists $\ell$ in range $1 \leq \ell \leq \lceil \log_2a
\rceil$ such that $ 2^{\ell -1} \leq |L| \leq s^{\ell}$. In other words, for $s = 2^{\ell}$, we have $s/2 \leq |L| \leq s$.  Since $x$ is independently and uniformly sampled, the probability that $x \in L$
is $|L|/n$, which is at least $\geq s/(2n)$.  Therefore,  by sampling $O(n/s)$ rounds, w.h.p. there is at least one event where $x \in L$.  

From now we assume that $|L| \leq s$ and that $x \in L$. 
We show that $\vol_G^{\out}(L) = s(s+k)$. Since $|L| \leq s$, $\vol_G^{\out}(L) = |E_G(L, L)| +  |E_G(L, S)| \leq |L|^2 + |L| |S| \leq s^2 + sk $. 

We have verified the following conditions for the parameters $x, \nu,
k$ for LocalVC$(G,x,\nu,k)$:
\begin{itemize}[nolistsep, noitemsep]
\item Local conditions are satisfied by \Cref{ass:correctness}\ref{item:correct_three}. 
\item $x \in L$. 
\item $|S| \leq k$.
\item $\vol_G^{\out}(L) \leq \nu$ since $\vol_G^{\out}(L) \leq s^2 + sk$, and we use $\nu = s^2+sk$. 
\end{itemize}
  By \Cref{thm:local-vertex-connectivity}, LocalVC outputs a vertex-cut of size at most $(1+\epsilon)k$.
\end{proof}

The final case is when every separation triple $(L,S,R)$ in $G$, $|S| >
k$. In other words, $\kappa_G > k$. If \Cref{alg:vcframework} outputs a vertex-cut, then it is a $(1+\epsilon)$-approximate vertex-cut. Otherwise,  \Cref{alg:vcframework} outputs $\perp$ correctly. 
The proof of \Cref{lem:vertex_sampling_nolocalvc_correct} is complete since
\Cref{lem:vertex_sampling_nolocalvc_correct} follows from
\Cref{lem:balanced_v}, \Cref{lem:imbalanced_v}, and the case $\kappa_G  > k$ corresponding the three cases of the set of separation triples in $G$.

\subsubsection{Vertex-Sampling without LocalVC}

We do not specify $k$ and LocalVC algorithm. 

\begin{lemma} \label{lem:vertex_sampling_nolocalvc_correct}
 \Cref{alg:vcframework} with vertex-sampling, but without  $\operatorname{LocalVC }$ w.h.p. outputs   a vertex-cut of size $\leq (1+\epsilon)\kappa_G$ 
\end{lemma}
\begin{proof}
Let $\tilde \kappa$ be the answer of our algorithm. By design, we have $\tilde \kappa \leq \min(d_{\min}^{\out}, d_{\min}^{\textin})$. Also, $\tilde \kappa \geq \kappa$ since the answer corresponds to some vertex-cut. It remains to show $\tilde \kappa \leq (1+O(\epsilon))\kappa$. 

Let $(L,S,R)$ be an optimal separation triple. We assume without loss
of generality that $|L| \leq |R|$. The other case is symmetric, where
we use $d_{\min}^{\textin}$ instead.

 We first show the inequality $|L| \geq  d_{\min}^{\out} - \kappa$. Since $(L,S,R)$ is a  separation triple where $|S| = \kappa$, the number of out-neighbors of a fixed vertex $x \in L$ that can be included in $S$ is at most $\kappa$. By definition of separation triple, neighbors of $x$ cannot be in $R$, and so the rest of the neighbors must be in $L$.

If $|L| \leq \epsilon d_{\min}^{\out}$, then $\kappa=|S|\ge d_{\min}^{\out}-\epsilon  d_{\min}^{\out} \ge\tilde{\kappa}(1-\epsilon)\ge\kappa(1-\epsilon)$.
That is, $\tilde{\kappa}$ is indeed an $(1+O(\epsilon))$-approximation
of $\kappa$ in this case.

On the other hand, if $|L| \geq \epsilon d_{\min}^{\out}$, then we claim that $|R| \geq \epsilon n/4$. To see this, if $d_{\min}^{\out} \geq n/2$, then $|R| \geq |L|   \geq \epsilon d_{\min}^{\out} \geq \epsilon  n/2$. Otherwise, $d_{\min}^{\out} \leq n/2$. In this case, $\kappa \leq d_{\min}^{\out} \leq n/2$. Therefore, $2|R| \geq |L|+|R| = n - |S| = n - \kappa \geq n/2$. In either case, the claim follows. 

We show that the probability that two sample vertices $x \in L$ and $y \in R$ is at least $\epsilon^2  d_{\min}^{\out}/(4n)$. First of all, the two events are independent. Recall that $|L| \geq \epsilon d_{\min}^{\out}$ and  $|R| \geq \epsilon n/4$. Therefore, $P(x \in L, y \in R) = P(x \in L)P(y \in R) = (|L|/n)(|R|/n) \geq \epsilon^2 d_{\min}^{\out}/(4n)$. 

Therefore, we sample for $\ot(n/(\epsilon^2d_{\min}^{\out} ))$ many times to get the event $x \in L$ and $y \in R$ w.h.p. Hence, we compute approximate $\kappa(x,y)$ correctly, and so our answer $\tilde \kappa$ is indeed an $(1+\epsilon)$-approximation.
\end{proof}

\subsection{Running Time} 

Let $T_1(m,n, k,\epsilon)$ be the time for deciding  if $\kappa(x,y) \leq (1+\epsilon)$,  $T_2(\nu, k, \epsilon)$ be the
running time for approximate  LocalVC, and $T_3(m,n,\epsilon)$ be the time for computing approximate  $\kappa(x,y)$.  If $G$ is undirected, we can replace $m$ with $nk$ with additional $O(m)$ preprocessing time. The running time for exact version is similar except that we do not have to pay $1/\epsilon$ factor for the first loop of \Cref{alg:vcframework}. 

\subsubsection{Edge-Sampling with LocalVC}
\begin{lemma} \label{lem:edge_sampling_localvc_time}
 \Cref{alg:vcframework} with edge-sampling, and
 $\operatorname{LocalVC }$  terminates in time $$\ot( (m/(\epsilon a))(T_1(m,n,
 k,\epsilon) + T_2(a, k, \epsilon) )).$$ 
\end{lemma}
\begin{proof}
 The first term comes from the first loop of \Cref{alg:vcframework}. That is, we repeat $O(m/(a\epsilon))$ times for computing approximate $\kappa(x,y)$, and each iteration takes  $T_1(m,n, k,\epsilon)$ time.  

The second term comes from computing local vertex connectivity.  For each $s \in \cL$, we repeat the second loop for $O(m/s)$ times, each LocalVC subroutine takes  $T_2(\nu, k, \epsilon)$ time where $\nu = s$. Therefore, the total time for the second loop  is $\sum_{s \in \cL}(m/s)T_2(s, k, \epsilon) = \ot( (m/a)T_2(a, k, \epsilon ))$. 
\end{proof}

\subsubsection{Vertex-Sampling with LocalVC}
\begin{lemma} \label{lem:vertex_sampling_localvc_time}
 \Cref{alg:vcframework} with vertex-sampling, and
 $\operatorname{LocalVC }$ terminates in time $$\ot( (n/(\epsilon a))(T_1(m,n, k,\epsilon) + T_2(a^2+ak, k, \epsilon) )).$$  
\end{lemma}
\begin{proof}
 The first term comes from the first loop of \Cref{alg:vcframework}. That is, we repeat $O(n/(a\epsilon))$ times for computing approximate $\kappa(x,y)$, and each iteration takes  $T_1(m,n, k,\epsilon)$ time.  

 The second term comes from computing local vertex connectivity.  For each $s \in \cL$, we repeat the second loop for $O(n/s)$ times, each LocalVC subroutine takes  $T_2(\nu, k, \epsilon)$ time where $\nu = O(s(s+k))$. Therefore, the total time for the second loop  is $\sum_{s \in \cL}(n/s)T_2(s, k, \epsilon) = \ot( (n/a)T_2(a^2 +ak, k, \epsilon ))$. 
\end{proof}

\subsubsection{Vertex-Sampling without LocalVC}
\begin{lemma} \label{lem:vertex_sampling_nolocalvc_time}
 \Cref{alg:vcframework} with vertex-sampling, but without
 $\operatorname{LocalVC }$ terminates in time $$\ot(  n/(\epsilon^2 k )T_3(m,n,\epsilon)).$$
\end{lemma}
\begin{proof}
The running time follows from the first loop where  we set $a$ such
that the number of sample is $n/(\epsilon^2 k )$, and computing approximate $\kappa(x,y)$ can be done in $T_3(m,n,\epsilon)$ time.
\end{proof}

\subsection{Proof of \Cref{thm:exact_vertex_connectivity,thm:approx_vertex_connectivity}}

For exact vertex connectivity, LocalVC runs in  $\nu^{1.5}k$
time by \Cref{cor:exact-local-vertex-connectivity}.  We can decide $\kappa(x,y) \leq k$ in $O(mk)$ time. 

For undirected exact vertex connectivity where $k < O(\sqrt{n})$, we first sparsifiy the graph
in $O(m)$ time. Then, we use edge-sampling with LocalVC algorithm
where we set $a = m'^{2/3}$, where $m' = O(nk)$ is the number of edges of sparsified graph. 

For directed exact vertex connectivity where $k < O(\sqrt{n})$, we use edge-sampling with
LocalVC algorithm where we set $a = m^{2/3}$ if $m < n^{3/2}$. If $m >
n^{3/2}$, we use vertex-sampling with LocalVC algorithm where we set
$a = m^{1/3}$. 

For approximate vertex connectivity, approximate LocalVC runs in $
\poly(1/\epsilon) \nu^{1.5}/ \sqrt{k}$ by
\Cref{thm:local-vertex-connectivity}. Also, we can decide $\kappa(x,y)
\leq (1+O(\epsilon))k$ or cerify that $\kappa \geq k$ in  time
\\$\ot(\poly(1/\epsilon) \min(mk,n^{2+o(1)}))$. The running time
$\poly(1/\epsilon) n^{2+o(1)} $ is due to \cite{ChuzhoyK19}. 

For undirected approximate vertex connectivity,  we first sparsify the
graph in $O(m)$ time. Let $m'$ be the number of edges of the
sparsified graph.  For $k < n^{0.8}$, we use edge-sampling with approximate LocalVC
algorithm where we set $a = m^{\hat a}$, where $\hat a = \frac{\min(5\hat k +
  2, \hat k +4)}{ 3\hat k + 3}$, and $\hat k = \log_nk$. For $k > n^{0.8}$, we use
vertex-sampling without LocalVC. 

For directed approximate vertex connectivity,  If $k \leq \sqrt{n}$ , we run edge-sampling with $a = m^{\hat a}, \hat a =
\min(2/3+\hat k, 1)$ where $\hat k = \log_m k$, or we run
vetex-sampling with $a =  m^{1/3} k^{1/2}$. If $ \sqrt{n} < k \leq
n^{0.8}$, we run edge-sampling with $a = m^{\hat a}$ where $\hat a =
4\log_mn/3 + \log_mk/3$ or vertex-sampling with $a = n^{\hat a}$ where
$\hat a = (2/3+ (\log_nk)/6)$. Finally, if $k > n^{4/5}$, we use vertex-sampling without
LocalVC.

	\section{$(1+\epsilon)$-Approximate Vertex Connectivity via Convex Embedding}

\begin{theorem}
There exists an algorithm that takes $G$ and $\epsilon > 0$, and in
$O(n^{\omega}/ \epsilon^2 + \min(\kappa_G, \sqrt{n})m )$ time  outputs a
vertex-cut $U$ such that $|U| \leq (1+\epsilon)\kappa$. 
\end{theorem}

\subsection{Preliminaries} 

\begin{definition}[Pointset in $\mathbb{F}^k$]
Let $\mathbb{F}$ be any field. For $k \geq 0$, $\mathbb{F}^k$ is
$k$-dimensional linear space over $\mathbb{F}$.  Denote $X = \{ x_1,
\ldots, x_n \}$ as a finite set of points in $\mathbb{F}^k$. The
\textit{affline hull} of $X$ is aff$(X) = \{  \sum_{i=1}^k c_ix_i
\text{ | } x_i \in X \text{ and }   \sum_{i=1}^k c_i = 1\}$. The rank
of $X$ denoted as rank$(X)$ is one plus dimension of aff$(X)$. In particular, if
$\mathbb{F} = \mathbb{R}$, then we will consider the \textit{convex hull} of $X$, denoted as conv$(X)$.
\end{definition}

For any sets $V, W $, any funtion $f : V \rightarrow W$, and any subset $U \subseteq V$,
we denote $f(U) = \{ f(u) \text{ | } u \in U \}$.  

\begin{definition}[Convex directed $X$-embedding] 
For any $X \subset V$, a convex directed $X$-embedding of
a graph $G = (V,E) $ is a function $f : V  \rightarrow \mathbb{R}^{|X|-1}$
such that for each $v \in V \setminus X$, $f(v) \in \text{conv}(f(N_G^{\text{out}}(v)))$.                         
\end{definition}
 
For efficiency point of view, we use the same method from \cite{LinialLW88,CheriyanR94}
that is based on convex-embedding over finite field $\mathbb{F}$. In
particular, they construct the directed $X$-embedding over the field of integers
modulo a prime $p$, $\mathbb{Z}_p$ by fixing a random prime number $p
\in [n^5, n^6]$, and choosing a random nonzero coefficient function
$c: E \rightarrow (\mathbb{Z}_p \setminus \{ 0 \})$ on edges. This
construction yields a function $f : V \rightarrow
(\mathbb{Z}_p^{|X|-1})$ called \textit{random modular directed
  $X$-embedding}.

\begin{definition} 
For $X,Y \subseteq V$, $p(X,Y)$ is the maximum number of vertex-disjoint paths
from $X$ to $Y$ where different paths have different end points. 
\end{definition}

\begin{lemma} \label{lem:embedding}
For any non-empty subset $U \subseteq V \setminus X$, w.h.p. a random modular directed
$X$-embedding $f : V \rightarrow \mathbb{Z}_p^{|X|-1}$ satisfies
$\rank(f(U)) = p(U,X)$. 
\end{lemma}

\begin{definition} [Fixed $k$-neighbors]
For $v \in V$, let $N_{G,k}^{\text{out}}(v)$ be a fixed, but arbitrarily selected
subset of $N_{G}^{\text{out}}(v)$ of size
$k$. Similarly, For $v \in V$, let $N_{G,k}^{\text{in}}(v)$ be a fixed, but arbitrarily selected
subset of $N_{G}^{\text{in}}(v) $ of size $k$. 
\end{definition}

\begin{lemma}  \label{lem:embeddingtime}Let $\omega$ be the exponent of the running time of the
  optimal matrix multiplication algorithm. Note it is known that $\omega \leq 2.372$. 
\begin{itemize}
\item For $y \in V$, a random modular directed
  $N_{G,k}^{\out}(y)$-embedding $f$ can be constructed in $O(n^{\omega})$
  time. 
\item Given such $f$, for $U \subseteq V$ with $|U| = k$, $\rank (f(U))$
  can be computed in $O(k^{\omega})$ time. 
\end{itemize}
\end{lemma}

\begin{lemma} [\cite{Gabow06}] \label{lem:gabow-shore}
For any optimal out-vertex shore $S$ such that $|N_G^{\out}(S)| =
\kappa_G$, then $\kappa_G \geq d_{\text{min}}^{\out} - |S|$. 
\end{lemma}

For any set $S,S'$, we denote $\min(S,S')$ as the set with smaller
cardinality.  

\subsection{Algorithm} 

\begin{algorithm}[H]
\KwIn{$G = (V,E)$, and $\epsilon > 0$} 
\KwOut{A vertex-cut $U$ such that w.h.p. $|U| \leq (1+\epsilon)
  \kappa_G$.} 
   
\BlankLine
Let $k \leftarrow \max(d_{\text{min}}^{\out},
d_{\text{min}}^{\textin})$. \;
Let $k' \leftarrow \min(d_{\text{min}}^{\out}, d_{\text{min}}^{\textin})$. 

\Repeat{ $\Theta (1/\epsilon)$ \normalfont{times}}
{
  Sample two random vertices $x_2,y_1 \in V$. \;
  Let $f$  be a random modular directed
  $N_{G,k}^{\textin}(y_1)$-embedding.  \tcp*{$O(n^{\omega})$ time.}
  Let $f^R$ be a random modular directed  $N_{G^R,k}^{\textin}(x_2)$-embedding. \;
  \Repeat{ $\Theta (n/(\epsilon k'))$ \normalfont{times}} 
  {
     Sample two random vertices $y_2,x_1 \in V$. \;
     $\rank (x_1, y_1) \leftarrow \rank (f(N_{G,k}^{\out}(x_1)))$ \tcp*{$O(k^{\omega})$ time.}
     $\rank (x_2, y_2) \leftarrow \rank (f^R(N_{G^R,k}^{\out}(y_2)))$ \;
  }
}
Let $x^*,y^*$ be the pair of vertices with minimum $\rank(x,y)$ for all $x,y$ computed so far.\;
Let $W \leftarrow \min(\kappa_G(x^*,y^*), \kappa_{G^R}(x^*,y^*))$ \;
Let $v_{\text{min}}, u_{\text{min}} $ be the vertex with the minimum out-degree in $G$ and $G^R$ respectively. \;
\Return{ \normalfont{min}$(W,
|N_G^{\out}(v_{\text{min}})|,|N_{G^R}^{\out}(u_{\text{min}})|) $}
\tcp*{Vertex-cut with minimum cardinality.}
\caption{ApproxConvexEmbedding$(G, \epsilon)$}
\label{alg:approxconvexembedding}
\end{algorithm}

\subsection{Analysis}

\global\long\def\doutmin{d_{\min}^{\out}}
\global\long\def\dinmin{d_{\min}^{\textin}}

\begin{lemma}
\Cref{alg:approxconvexembedding} outputs w.h.p. a vertex-cut $U$
such that $|U| \leq (1+\epsilon)\kappa_G$. 
\end{lemma} 
\begin{proof}

Let $\tilde{\kappa}$ denote the answer of our algorithm. Clearly
$\tilde{\kappa}\le\doutmin$ and $\tilde{\kappa}\le\dinmin$ by design.
Observe also that $\tilde{\kappa}\ge\kappa$ because the answer corresponds
to some vertex cut. Let $(A,S,B)$ be the optimal separation triple
where $A$ is a out-vertex shore and $|S|=\kappa$. W.l.o.g. we assume
that $|A|\le|B|$, another case is symmetric.

Suppose that $|A|\le\epsilon\doutmin$. Then $\kappa=|S|\ge\doutmin-\epsilon\doutmin\ge\tilde{\kappa}(1-\epsilon)\ge\kappa(1-\epsilon)$.
That is, $\tilde{\kappa}$ is indeed an $(1+O(\epsilon))$-approximation
of $\kappa$ in this case.

Suppose now that $|A|\ge\epsilon\doutmin$. We claim that $|B|\ge\epsilon n/4$.
Indeed, if $\doutmin\ge n/2$, then $|B|\ge\epsilon n/2$. Else if,
$\doutmin\le n/2$, then we know $|S|=\kappa\le n/2$. But $2|B|\ge|A|+|B|=n-|S|\ge n/2$.
In either case, $|B|\ge\epsilon n/4$. 

Now, as $|B|\ge\epsilon n/4$ and we sample $\tilde{O}(1/\epsilon)$
many $y_{1}$. There is one sample $y_{1}\in B$ w.h.p. and now we
assume that $y_{1}\in B$. In the iteration when $y_{1}$ is sampled.
As $|A|\ge\epsilon\doutmin$ and we sample at least $\tilde{O}(n/\doutmin\epsilon)$
many $x_{1}$. There is one sample $x_{1}\in A$ w.h.p.

  By \Cref{lem:embedding}, w.h.p.,   $$\rank(x_1,y_1) = \rank
  (f(N_{G,k}^{\out}(x_1))) = p(N_{G,k}^{\out}(x_1), N_{G,k}^{\textin}(y_1)
  ) = \kappa(x_1,y_1) = \kappa.$$ 
  So our answer $\tilde{\kappa} = \kappa$ in this case.

\end{proof}

\begin{lemma}
\Cref{alg:approxconvexembedding} terminates in $O(n^{\omega}/ \epsilon^2
+\min(\kappa_G, \sqrt{n})m )$ time.
\end{lemma}
\begin{proof}
By \Cref{lem:embeddingtime}, the construction time for a random modular directed
  $N_{G,k}^{\textin}(y_1)$-embedding is $O(n^{\omega})$.  Given $y_1$,
  we sample $x_1$ for $\Theta (n/(\epsilon k))$ rounds. Each round we can
  compute $\kappa(x_1,y_1)$ by computing
  $\rank(f(N_{G,k}^{\out}(x_1)))$  in $O( k^{\omega})$ time. Hence,
  the total time is to find the best pair $(x,y)$ is $O(n^{\omega} + n k^{\omega-1} / \epsilon) =
  O(n^{\omega}/\epsilon) $. Finally, we can compute $\kappa_G(x,y)$ to obtain the
  vertex-cut for the best pair in $O( \min(\kappa_G, \sqrt{n})m)$
  time. It takes linear time to compute $ N_G^{\out}(v_{\text{min}})$ and
  $N_{G^R}^{\out}(u_{\text{min}})$. Hence, the result follows.  
\end{proof}
          
	\section{Open Problems}
\label{sec:open}
\begin{enumerate}[nolistsep,noitemsep] 
\item Is there an $O(\nu k)$-time LocalVC algorithm?
\item Can we break the $O(n^3)$ time bound when $k = \Omega(n)$? This would still be hard to break even if we had an $O(\nu k)$-time LocalVC algorithm. 
\item Is there an $o(n^2)$-time algorithm for vertex-weighted graphs when $m=O(n)$?  Our LocalVC algorithm does not generalize to the weighted case. 
\item Is there an $o(n^2)$-time algorithm for the single-source max-flow problem when $m=O(n)$?
\item Is there a near-linear-time $o(\log n)$-approximation algorithm?
\item How fast can we solve the vertex connectivity problem in the dynamic setting (under edge insertions and deletions) and the distributed setting (e.g. in the CONGEST model)?
\end{enumerate}

	\section*{Acknowledgement}  
	This project has received funding from the European Research
        Council (ERC) under the European Union's Horizon 2020 research
        and innovation programme under grant agreement No
        715672 and 759557. Nanongkai was also partially supported by the Swedish
        Research Council (Reg. No. 2015-04659.).
        

	\ifdefined\AdvCite

	\printbibliography[heading=bibintoc] 
	
	\else 
 	
	\bibliographystyle{alpha}
	\bibliography{references} 
	
	\fi

\end{document}